%% file: main-RelaibleSpanner.tex
\documentclass[11pt,english]{article}
\usepackage{amsmath,amssymb,amsthm}
\usepackage{multicol}
\usepackage[margin=1in]{geometry}
\usepackage{graphicx,color}
\usepackage{enumitem}
\usepackage{fullpage}
\usepackage[noblocks]{authblk}
\usepackage{tcolorbox}
\usepackage{babel}
\usepackage{wrapfig}
\usepackage{MnSymbol}
\usepackage{mdframed}
\usepackage{mathtools}
\usepackage{floatrow}
\usepackage{multirow}
\usepackage{tablefootnote}

\usepackage[ruled,linesnumbered,vlined]{algorithm2e}
\usepackage{tablefootnote}
\usepackage{thm-restate}
\usepackage{bbding}
\usepackage{pifont}
\usepackage{nicefrac}
\usepackage{undertilde}

\definecolor{Darkblue}{rgb}{0,0,0.4}
\definecolor{Brown}{cmyk}{0,0.61,1.,0.60}
\definecolor{Purple}{cmyk}{0.45,0.86,0,0}
\definecolor{Darkgreen}{rgb}{0.133,0.543,0.133}

\usepackage[colorlinks,linkcolor=Darkblue,filecolor=blue,citecolor=blue,urlcolor=Darkblue,pagebackref]{hyperref}
\usepackage[nameinlink]{cleveref}

\usepackage[colorinlistoftodos,prependcaption,textsize=tiny]{todonotes}

\newif\ifdraft 
\draftfalse

\newcommand{\namedref}[2]{\hyperref[#2]{#1~\ref*{#2}}}
\newcommand{\propref}[1]{\hyperref[#1]{property~(\ref*{#1})}}

\newcommand{\mthmref}[1]{\namedref{meta Theorem}{#1}}

\newtheorem{theorem}{Theorem}
\newtheorem{lemma}{Lemma}
\newtheorem{definition}{Definition}

\newtheorem{observation}{Observation}
\newtheorem{corollary}{Corollary}

\newtheorem{remark}{Remark}

\newcommand{\R}{\mathbb{R}}
\newcommand{\N}{\mathbb{N}}

\newcommand{\supp}{\mathrm{supp}}

\newcommand{\diam}{\mathrm{diam}}

\newcommand{\SPD}{\textsf{SPD}\xspace}

\newcommand{\SPDdepth}{\textsf{SPDdepth}\xspace}
\newcommand{\lca}{\mathsf{lca}}

\definecolor{forestgreen}{rgb}{0.13, 0.55, 0.13}

\SetCommentSty{mycommfont}

\def\eps{\epsilon}

\DeclareMathAlphabet{\mathpzc}{OT1}{pzc}{m}{it}
\newcommand{\etal}{{\em et al. \xspace}}

\newlength{\dhatheight}

\newcommand {\ignore} [1] {}

\newcommand{\initOneLiners}{%
	\setlength{\itemsep}{0pt}
	\setlength{\parsep }{0pt}
	\setlength{\topsep }{0pt}
}

\title{Locality-Sensitive Orderings and Applications to Reliable Spanners\thanks{This paper appeared in \href{https://arxiv.org/abs/2101.07428}{arXiv} under the title ``Reliable Spanners: Locality-Sensitive Orderings Strike Back''. We changed the title in order to emphasize that the main theme of the paper is Locality-Sensitive Orderings.}}

\author{Arnold Filtser\thanks{	Email: \texttt{arnold273@gmail.com}.}}
\affil{Bar-Ilan University}
\author{Hung Le\thanks{	Email: \texttt{hungle@cs.umass.edu}. The research was supported by the start-up grant of Umass Amherst and by the National Science Foundation under Grant No. CCF-2121952.}}
\affil{University of Massachusetts at Amherst}

\date{}
\begin{document}
\maketitle
\begin{abstract}
Chan, Har-Peled, and Jones [2020] recently developed locality-sensitive ordering (LSO), a new tool that allows one to reduce problems in the Euclidean space $\mathbb{R}^d$ to the $1$-dimensional line. They used LSO's to solve a host of problems. 
Later, Buchin, Har-Peled, and Ol{\'{a}}h [2019,2020] used the LSO of Chan {\em et al. } to construct very sparse \emph{reliable spanners} for the Euclidean space. A highly desirable feature of a reliable spanner is its ability to withstand a massive failure: the network remains functioning even if 90\% of the nodes fail. 
In a follow-up work, Har-Peled, Mendel, and Ol{\'{a}}h [2021] constructed reliable spanners for general and topologically structured metrics. Their construction used a different approach, and is based on sparse covers.
	
In this paper, we develop the theory of LSO's in non-Euclidean metrics by introducing new types of LSO's suitable for general and topologically structured metrics. 
We then construct such LSO's, as well as constructing considerably improved LSO's for doubling metrics.
Afterwards, we use our new LSO's to construct reliable spanners with improved stretch and sparsity parameters. 
Most prominently, we construct $\tilde{O}(n)$-size reliable spanners for trees and planar graphs with the optimal stretch of $2$.
Along the way to the construction of LSO's and reliable spanners, we introduce and construct ultrametric covers, and construct $2$-hop reliable spanners for the line.

\end{abstract}

\newpage

		\setcounter{tocdepth}{2} \tableofcontents

    \newpage
    \pagenumbering{arabic}

\section{Introduction}
The Algorithmist's toolkit consists of diverse ``tools'' frequently utilized for many different problems. In the geometric context, some tools apply to general metric spaces such as metric embeddings \cite{Bou85,FRT04} and padded decompositions \cite{KPR93,Bar96,Fil19padded}, while many tools apply mainly to Euclidean spaces, such as dimension reduction \cite{JL84}, locality-sensitive hashing \cite{IM98}, well-separated pair decomposition (WSPD) \cite{CK95}, and many others.  Recently, Chan, Har-Peled, and Jones \cite{CHJ20} developed a new and exciting tool for Euclidean spaces called \emph{Locality-Sensitive Ordering} (LSO). 
\begin{restatable}[$(\tau,\rho)$-LSO]{definition}{defLSO}
	\label{def:LSOclassic}
	Given a metric space $(X,d_{X})$, we say that a collection $\Sigma$
	of orderings is a $(\tau,\rho)$-LSO (locality-sensitive ordering) if
	$\left|\Sigma\right|\le\tau$, and for every $x,y\in X$, there is
	a linear ordering $\sigma\in\Sigma$ such that (w.l.o.g.) $x\prec_{\sigma}y$ and the points between $x$ and $y$ w.r.t. $\sigma$ could be partitioned
	into two consecutive intervals $I_{x},I_{y}$ where $I_{x}\subseteq B_{X}(x,\rho\cdot d_{X}(x,y))$	and $I_{y}\subseteq B_{X}(y,\rho\cdot d_{X}(x,y))$. Parameter $\rho$ is called the \emph{stretch} parameter. 
\end{restatable}

\sloppy The main reason that LSO has become an extremely useful tool is that it reduces the problem at hand in the $d$-dimensional Euclidean space to the same problem in a much simpler space: the $1$-dimensional line. 
\cite{CHJ20} constructed an $(O(\epsilon)^{-d}\log\frac{1}{\epsilon},\eps)$-LSO for any given set of points in the $d$-dimensional Euclidean space $\mathbb{R}^d$ (more generally, Chan \etal \cite{CHJ20} constructed $\left(O(\epsilon^{-1}\cdot \log n)^{O(d)},\eps\right)$-LSO for metric spaces with doubling dimension\footnote{A metric space $(X, d)$ has doubling dimension $d$ if every ball of radius $2r$ can be 	covered by $2^{d}$ balls of radius $r$.\label{foot:doubling}} $d$). 
They used their LSO to design simple dynamic algorithms for approximate nearest neighbor search, approximate bichromatic closest pair, approximate MST, spanners, and fault-tolerant spanners.  Afterwards, Buchin, Har-Peled, and Ol{\'{a}}h \cite{BHO19,BHO20} used the LSO of  Chan \etal \cite{CHJ20} to construct \emph{reliable spanners} (see \Cref{def:reliableSpanner}) for Euclidean spaces following the same methodology: reducing the problem to the construction on the line. In this work, we introduce new notions of LSO and apply them to construct reliable spanners for non-Euclidean metrics.

Given a metric space $(X,d_X)$, a $t$-spanner is a weighted graph $H=(X,E,w)$ over\footnote{Often in the literature, the metric space $(X,d_X)$ is the shortest path metric of a graph $G$, and there is a requirement that $H$ will be a subgraph of $G$. We will not have such a requirement in this paper.} $X$ where for every pair of points $x,y\in X$, $d_X(x,y)\le d_H(x,y)\le t\cdot d_X(x,y)$, with $d_H$ being the shortest path metric of $H$. The parameter $t$ is called the stretch of the spanner.  A highly desirable property of a $t$-spanner is the ability to withstand extensive vertex failures. Levcopoulos, Narasimhan, and Smid \cite{LNS02} introduced the notion of a fault-tolerant spanner.  A subgraph $H=(V,E_H,w)$ is an $f$-\emph{vertex-fault-tolerant} $t$-spanner of a weighted graph $G=(V,E,w)$, if for every set $F\subset V$ of at most $f$ vertices, it holds that $\forall u,v\notin F$, $d_{H\setminus F}(u,v)\le t\cdot d_{G\setminus F}(u,v)$.  A major limitation of fault-tolerant spanners is that the number of failures must be determined in advance; in particular, such spanners cannot withstand a massive failure.
One can imagine a scenario where a significant portion (even 90\%) of a network fails and ceases to function (due to, e.g., close-down during a pandemic), it is important that the remaining parts of the network (or at least most of it) will remain highly connected and functioning.  To this end, Bose \etal \cite{BDMS13} introduced the notion of a \emph{reliable spanner}. Here, given a failure set $B\subseteq X$, the residual spanner $H\setminus B$ is a $t$-spanner for $X\setminus B^+$, where $B^+\supseteq B$ is a set slightly larger than $B$.  
Buchin \etal \cite{BHO20} relaxed the notion of reliable spanners by allowing the size of $B^+$ to be bounded only in expectation.
 
\begin{definition}[Reliable spanner]\label{def:reliableSpanner}
	A weighted graph $H$ over point set $X$ is a deterministic $\nu$-reliable $t$-spanner
	of a metric space $(X,d_{X})$ if $d_{H}$ dominates 
	\footnote{Metric space $(X,d_H)$ dominates metric space $(X,d_X)$ if  $\forall u,v\in X$, $d_X(u,v)\le d_H(u,v)$.\label{foot:dominating}} 
	$d_{X}$, and for every
	set $B\subseteq X$ of points, called an \emph{attack set}, there is a set $B^{+}\supseteq B$, called a \emph{faulty extension} of $B$, 
	such that:
	\begin{enumerate}
		\item $|B^{+}|\le(1+\nu)|B|$.
		\item For every $x,y\notin B^{+}$, $d_{H[X\setminus B]}(x,y)\le t\cdot d_{X}(x,y)$.
	\end{enumerate}	
	An oblivious $\nu$-reliable $t$-spanner is a distribution $\mathcal{D}$ over dominating graphs $H$, such that for every attack set $B\subseteq X$ and $H\in\supp(\mathcal{D})$,
	there 
	exist a superset $B^{+}$ of $B$ such that, for
	every $x,y\notin B^{+}$, $d_{H[X\setminus B]}(x,y)\le t\cdot d_{X}(x,y)$,
	and  $\mathbb{E}_{H\sim\mathcal{D}}\left[|B^{+}|\right]\le(1+\nu)|B|$. We say that the oblivious spanner $\mathcal{D}$ has $m$ edges if every graph $H\in\supp(\mathcal{D})$ has at most $m$ edges.
\end{definition}

We call the distribution $\mathcal{D}$ in \Cref{def:reliableSpanner} an oblivious $\nu$-reliable $t$-spanner because the adversary is oblivious to the specific spanner produced by the distribution (it may be aware to the distribution itself).

For constant dimensional Euclidean spaces, Bose \etal \cite{BDMS13} constructed a deterministic reliable $O(1)$-spanner, such that for every attack $B$, the faulty extension $B^+$ contains at most $O(|B|^2)$ vertices. The construction of reliable spanners where the size of $B^+$ is a linear function of $B$ was left as an open question.
For every $\nu,\eps\in(0,1)$, and $n$ points in $d$-dimensional Euclidean space $(\R^d,\|\cdot\|_2)$, Buchin \etal \cite{BHO19} used the LSO of Chan \etal \cite{CHJ20}  to construct a  \emph{deterministic} $\nu$-reliable $(1+\eps)$-spanner with $n\cdot\nu^{-6}\cdot\tilde{O}(\epsilon)^{-7d}\cdot\tilde{O}(\log n)$ 
edges (see also  \cite{BCDM18}).
Later, for the \emph{oblivious} case, Buchin \etal \cite{BHO20} applied the same LSO to construct an oblivious  $\nu$-reliable $(1+\eps)$-spanner with 
$n\cdot\tilde{O}(\epsilon)^{-2d}\cdot\tilde{O}(\nu^{-1}(\log\log n)^{2})$ edges.

Very recently, Har-Peled, Mendel, and Ol{\'{a}}h \cite{HMO21} constructed reliable spanners for general metric spaces, as well as for topologically structured spaces (e.g. trees and planar graphs). They  showed that for every integer $k$, every general $n$-point metric space admits an \emph{oblivious} $\nu$-reliable $(512\cdot k)$-spanner
\footnote{\cite{HMO21} did not compute the constant explicitly. Their construction is based on the Ramsey-trees of Mendel-Naor \cite{MN07}, which have stretch $128k$. Using state of the art Ramsey trees  \cite{NT12} of stretch $2ek$ instead (see also \cite{ACEFN20}), the approach of \cite{HMO21}  provides stretch $8ek$.\label{foot:MN07}}
with $n^{1+1/k}\cdot O(\nu^{-1}k\log^{2}\Phi\log\frac{n}{\nu})$ edges, where $\Phi=\frac{\max_{x,y}d_X(x,y)}{\min_{x,y}d_X(x,y)}$ is the \emph{aspect ratio} of the metric space (also known as the \emph{spread}, which a priori is unbounded). Additionally, they showed that ultrametrics (see \Cref{def:HST}) admit oblivious $\nu$-reliable $(2+\eps)$-spanners with $n\cdot \tilde{O}(\nu^{-1}\epsilon^{-2}\log^{2}\Phi)$ edges,
tree metrics admit oblivious $\nu$-reliable $(3+\eps)$-spanners with $n\cdot \tilde{O}(\nu^{-1}\epsilon^{-2}\log^2n\log^{2}\Phi)$ edges,
and planar metrics admit oblivious $\nu$-reliable $(3+\eps)$-spanners with $n\cdot \tilde{O}(\nu^{-1}\epsilon^{-4}\log^{2}\Phi)$ edges (see \Cref{tab:results}). 

The reliable spanner constructions of Har-Peled \etal \cite{HMO21} are based on \emph{sparse covers}. A $(\tau,\rho)$-sparse cover is a collection $\mathcal{C}$ of clusters such that every point belongs to at most $\tau$ clusters, and for every pair $x,y\in X$, there is a cluster $C\in\mathcal{C}$ containing both $x,y\in C$ where $\diam(C)\le \rho\cdot d_X(x,y)$; $\rho$ is called the \emph{stretch} of the cover $\mathcal{C}$. They then treat each cluster in $\mathcal{C}$ as a uniform metric, construct a reliable spanner for each cluster, and return the union of all the constructed spanners. 
Thus the main task becomes constructing a reliable spanner for the uniform metric.
Specifically, instead of the oblivious $\nu$-reliable $1$-spanner for the line constructed in \cite{BHO20}, Har-Peled \etal \cite{HMO21} constructed an oblivious $\nu$-reliable $2$-spanner for the uniform metric, which is the best stretch possible for subquadratic  size spanners (see \Cref{thm:LBtree}).
Indeed, this additional factor $2$ appears in the stretch parameter in all the spanners in \cite{HMO21}. Most prominently, for trees they constructed an $(O(\eps^{-1}\log\Phi\log n),2+\eps)$-sparse cover, resulting in a stretch $4+\eps$ spanner,~\footnote{With an additional effort, \cite{HMO21} reduced the stretch of the spanner to $3+\eps$. This analysis is tight, and their technique cannot give a reliable spanner with a stretch factor smaller than $3$.} while the natural lower bound is $2$ (\Cref{thm:LBtree}).
A similar phenomenon occurs for planar graphs.
An additional drawback in the sparse cover based approach of \cite{HMO21} is its dependency on the aspect ratio $\Phi$ (which a priori can be unbounded). This dependency on the aspect ratio is inherent in their technique and cannot be avoided (see Lemma 20 in \cite{HMO21}). 

\subsection{Our contribution}
Our major contribution is to the theory of locality-sensitive orderings. Specifically, we significantly improve the parameters of LSO in doubling metrics $^{\ref{foot:doubling}}$, and extend the idea of LSO to general metrics, as well as to topologically structured metrics. This is done by introducing left-sided LSO and triangle-LSO (see \Cref{tab:LSO}). 
LSO's are a powerful tool enabling one to reduce many problems to the line. LSO's already have  many applications in computational geometry~\cite{CHJ20}; we expect that our LSO for doubling metrics, as well as those for general and topologically structured graphs, will find many additional applications in the future.
Next, we use these newly introduced LSO's (or improved in the case of doubling) to construct oblivious reliable spanners. Our constructions have smaller stretch (optimal in the case of topologically structured metrics) and smaller sparsity (see \Cref{tab:results}). Below we describe each type of LSO in detail, and which spanners it was used to construct.
Our constructions of LSO for general and doubling metrics are going through the construction of \emph{ultrametric covers}. An ultrametric cover is a collection of dominating ultrametrics such that the distance between every pair of points is well approximated by some ultrametric in the collection. 
We construct the first ultrametric cover for doubling metrics with stretch $1+\eps$ (previously only tree covers were known), and improve the stretch parameter in the ultrametric covers of general metrics (see \Cref{tab:TreeCover}).
Finally, a crucial ingredient when one constructs reliable spanners using LSO is reliable spanners for the line. Buchin \etal \cite{BHO19,BHO20} constructed such spanners; however their spanners have $\Omega(\log n)$ hops, which will incur additional $\log n$ factor in the stretch (in all cases other than Euclidean/doubling). To avoid this overhead, we construct a $2$-hop reliable spanner for the line, and a $2$-hop left spanner, which is a newly defined type of spanner suitable for our left-LSO (see \Cref{tab:LineSpanners}).  
See \Cref{fig:concepts} for a graphic illustration of how the different parts in the paper are related.
Finally, we answer an open question by Har-Peled \cite{Har-Peled20} regarding sub-graph reliable spanners, by providing  matching upper and lower bounds for reliable connectivity preservers.

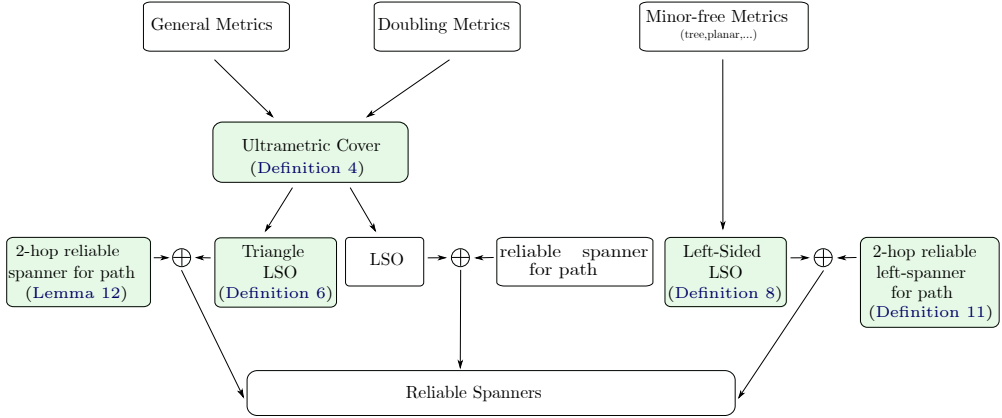
\begin{figure}[t]
	\centering
	\begin{picture}(400,170)
		\put(0,0){\includegraphics[width=0.8\textwidth]{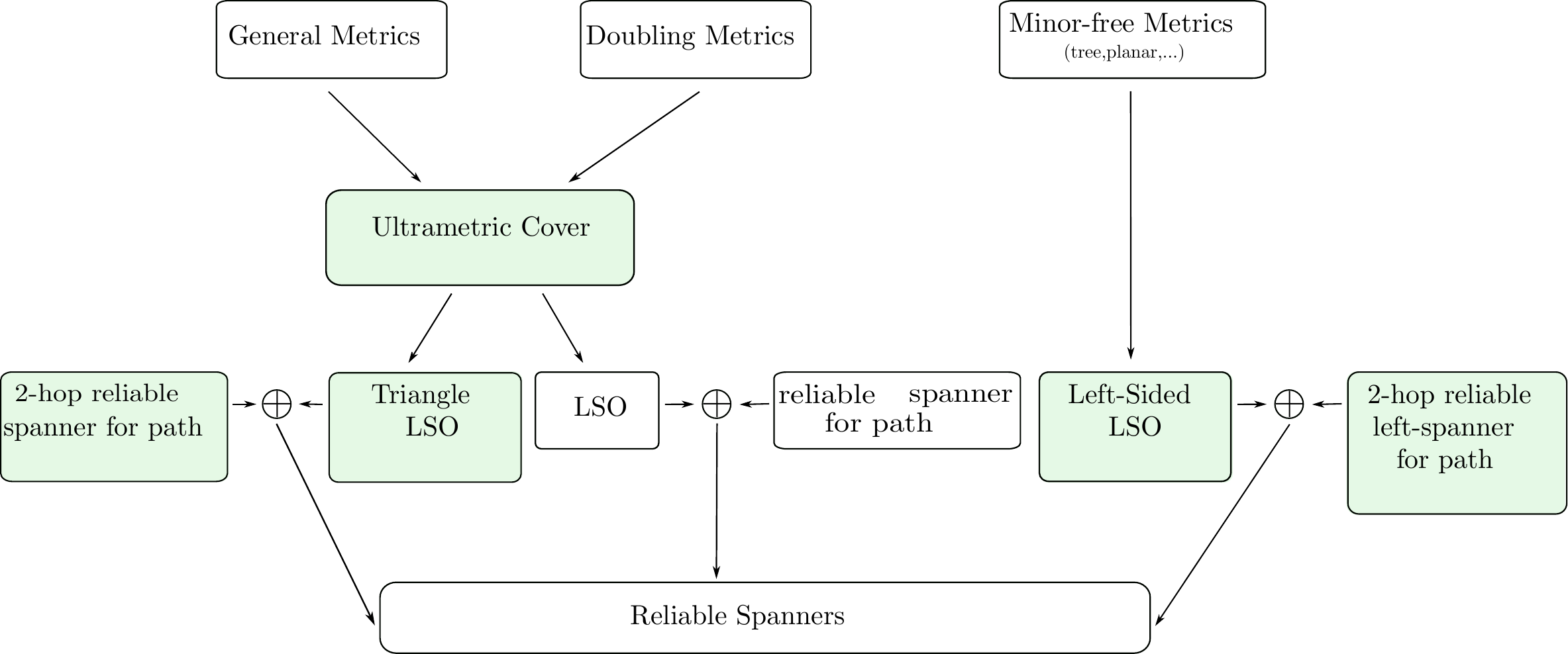}}
		\put(92,92){\tiny(\Cref{def:UltrametricCover})}
		\put(7,44.5){\tiny(\Cref{lem:2-hopSpanner})}
		\put(80,44.5){\tiny(\Cref{def:TriangleLSO})}
		\put(250,44.5){\tiny(\Cref{def:leftsidedLSO})}
		\put(326,37.4){\tiny(\Cref{def:RelLeftSpanner})}
	\end{picture}
	\caption{\footnotesize{Relationships between different concepts; new concepts introduced in this papers are green-shaded.}}
	\label{fig:concepts}
\end{figure}

\begin{table}[]
	\begin{tabular}{|l|l|l|l|l|}
		\hline
		LSO type                                             & Metric Space             & \# of orderings ($\tau$)                         & Stretch          & Ref                                                  \\ \hline
		\multicolumn{1}{|c|}{\multirow{3}{*}{(Classic) LSO}} & Euclidean space $\mathbb{R}^d$ & $O(\epsilon)^{-d}\cdot \log\frac{1}{\epsilon}$         & $\eps$           & \cite{CHJ20}                        \\ \cline{2-5} 
		\multicolumn{1}{|c|}{}                               & \multirow{2}{*}{Doubling dimension $d$}   & $O(\epsilon^{-1}\cdot \log n)^{O(d)}$                     & $\eps$           & \cite{CHJ20}                        \\ \cline{3-5} 
		\multicolumn{1}{|c|}{}                               &   & $\eps^{-O(d)}$                                   & $\eps$           & \Cref{cor:DoublingLSO}              \\ \hline
		\multirow{2}{*}{Triangle-LSO}                                        & General metric           & $\tilde{O}(n^{1/k}\cdot\eps^{-1})$ & $2k+\eps$ & \Cref{cor:TriangleLSOGeneralGraphs} \\ \cline{2-5}
		& Ultrametric           & $1$ & $1$ & \Cref{thm:CoverToTriangleLSO} \\ \hline		
		\multirow{4}{*}{Left-sided LSO}                      & Tree                     & $\log n$                                         & $1$              & \Cref{thm:LeftLSOTreewidth}         \\ \cline{2-5} 
		& Treewidth $k$            & $k\cdot \log n$                                        & $1$              & \Cref{thm:LeftLSOTreewidth}         \\ \cline{2-5} 
		& Planar graph             & $\frac{1}{\eps}\cdot \log^2 n$                         & $1+\eps$         & \Cref{thm:LeftLSOMinorFree}         \\ \cline{2-5} 
		& Minor Free               & $\frac{1}{\eps}\cdot \log^2 n$                         & $1+\eps$         & \Cref{thm:LeftLSOMinorFree}         \\ \hline
	\end{tabular}
	\caption{\small{Summery of all known results, on all the different types of locality sensitive orderings (LSO). $k\in\N$ is an integer, $\eps\in(0,1)$ is an arbitrarily small parameter.}}
	\label{tab:LSO}
\end{table}

\begin{table}[t]
	\begin{tabular}{|c|l|l|l|l|}
		\hline
		Family                                                                                    & stretch               & guarantee     & size                                                                                                                                 & ref                        \\ \hline
		\multirow{3}{*}{\begin{tabular}[c]{@{}l@{}}Euclidean\\ $(\R^d,\|\cdot\|_2)$\end{tabular}} & $O(1)$              & Deterministic    & $\Omega(n\log n)$
		& \cite{BDMS13}                     \\ \cline{2-5} 
		& $1+\eps$              & Deterministic & $n\cdot\tilde{O}(\epsilon)^{-7d}\nu^{-6}\cdot\tilde{O}(\log n)$
		& \cite{BHO19}                       \\ \cline{2-5} 
		& $1+\eps$              & Oblivious     & $n\cdot\tilde{O}(\epsilon)^{-2d}\cdot\tilde{O}(\nu^{-1}(\log\log n)^{2})$ & \cite{BHO20}          \\ \hline
		
		\multirow{2}{*}{\begin{tabular}[c]{@{}l@{}}Doubling\\ dimension $d$\end{tabular}}         & $1+\eps$              & Deterministic & $n\cdot\epsilon^{-O(d)}\nu^{-6}\cdot\tilde{O}(\log n)$
		&    \Cref{cor:Doubling}                        \\ \cline{2-5} 
		& $1+\eps$              & Oblivious     & $n\cdot\epsilon^{-O(d)}\nu^{-1}\log\nu^{-1}\cdot\tilde{O}(\log\log n)^{2}$
		&    \Cref{cor:Doubling}                        \\ \hline
		\multirow{2}{*}{\begin{tabular}[c]{@{}l@{}}General\\ metric\end{tabular}} 
		& $512\cdot k$ ~~ $^{\ref{foot:MN07}}$
		& Oblivious     & $\tilde{O}(n^{1+\nicefrac{1}{k}}\cdot\nu^{-1})\cdot\log^{2}\Phi$                                                       & \cite{HMO21}                       \\ \cline{2-5} 
		& $8k+\eps$ & Oblivious     & $\tilde{O}(n^{1+\nicefrac{1}{k}}\cdot\eps^{-2})\cdot\nu^{-1}$ & \Cref{thm:generalMetricOblivious} \\ \hline
		\multirow{2}{*}{\begin{tabular}[c]{@{}l@{}}Ultrametric,\\ Tree, planar \end{tabular}} & $k$                    & Deterministic         & $\Omega(n^{1+1/k})$  & \cite{HMO21}            \\\cline{2-5} 
		&$k<2$                    & Oblivious         & $\Omega(n^2)$  & \Cref{thm:LBtree}            \\ \hline	
		\multirow{2}{*}{Ultrametric}                                                              & $2+\eps$              & Oblivious     & $n\cdot \tilde{O}(\nu^{-1}\epsilon^{-2}\log^{2}\Phi)$                                            & \cite{HMO21}                       \\ \cline{2-5} 
		& $2$  {\scriptsize ~\qquad\color{red}(tight)}                  & Oblivious     & $n\cdot \tilde{O}\left(\log^{2}n+\nu^{-1}\log n\right)$                                                                  &  \Cref{cor:UltrametricSpanner}                          \\ \hline
		\multirow{2}{*}{Tree}                                                                     & $3+\eps$              & Oblivious     & $n\cdot \tilde{O}(\nu^{-1}\epsilon^{-2}\log^{2}n\log^{2}\Phi)$                             & \cite{HMO21}                       \\ \cline{2-5} 
		& $2$    {\scriptsize ~\qquad\color{red}(tight)}                & Oblivious     & $n\cdot O(\nu^{-1}\log^{3}n)$                                                                                                  & \Cref{thm:SpannerForTw}        \\ \hline
		Treewidth $k$                                                                            & $2$    {\scriptsize ~\qquad\color{red}(tight)}                & Oblivious     & $n\cdot O(\nu^{-1}k^{2}\log^{3}n)$                                                                                             & \Cref{thm:SpannerForTw}        \\ \hline
		\multirow{2}{*}{Planar}                                                                   & $3+\eps$              & Oblivious     & $n\cdot \tilde{O}(\nu^{-1}\epsilon^{-4}\log^{2}n\log^{2}\Phi)$                             & \cite{HMO21}                       \\ \cline{2-5} 
		& $2+\eps$   {\scriptsize ~~\color{red}(tight)}           & Oblivious     & $n\cdot O(\nu^{-1}\eps^{-2}\log^{5}n)$                                                                            & \Cref{thm:SpannerForMinorFree}        \\ \hline
		Minor-free    & $2+\eps$  {\scriptsize ~~\color{red}(tight)}                                                                                         & Oblivious     & $n\cdot O(\nu^{-1}\eps^{-2}\log^{5}n)$                                                                          & \Cref{thm:SpannerForMinorFree}        \\ \hline
	\end{tabular}
	\caption{\small{Comparison between previous and new constructions of reliable spanners. All spanners (except \cite{BDMS13}) constructed on $n$-point metric spaces with reliability parameter $\nu$. 
	For doubling metrics, we recover the same strong results previously known only for Euclidean space (up to a polynomial dependence on $\eps$). 
	Both lower bounds hold for the uniform metric (which is sub-metric of a star metric).
	For all other metric spaces,
	we improve both stretch and sparsity, and remove the undesirable dependence on the aspect ratio (spread) $\Phi$.
	Most notably, for trees and planar graphs, the stretch was improved from $3+\eps$, to the best possible stretch $2$. 
	Finally, for general graphs, our spanner has stretch $8k$, considerably improving the constant hiding in \cite{HMO21}. This constant is highly important as it governs the parameter in the power of $n$. 
		}}
	\label{tab:results}
\end{table}

\paragraph{Classic LSO.~} Chan \etal \cite{CHJ20} constructed a $\left((\eps^{-1}\cdot \log n)^{O(d)},\eps\right)$-LSO for metric spaces of doubling dimension $d$. Applying the (implicit) framework of \cite{BHO19,BHO20} yields a reliable spanner with $n\cdot(\eps^{-1}\cdot \log n)^{O(d)}$ edges.
In this work, we completely remove the dependency on $n$ of the number of orderings.  Specifically, we construct an $\left(\epsilon^{-O(d)},\epsilon\right)$-LSO for doubling metrics (\Cref{cor:DoublingLSO}); this immediately implies reliable spanners for metric spaces of doubling dimension $d$ with the same performance, up to the dependency on $\eps$, as for Euclidean spaces (\Cref{cor:Doubling}).

\paragraph{Triangle LSO.~} A $(\tau,\rho)$-triangle LSO for a metric space $(X,d_X)$ is a collection $\Sigma$ of at most $\tau$ linear orderings over $X$, such that for every $x,y\in X$, there is an ordering $\sigma\in\Sigma$ such that for every two points $z,q\in X$ satisfying  $x\preceq_{\sigma}z\preceq_{\sigma}q\preceq_{\sigma}y$ (or $y\preceq_{\sigma}z\preceq_{\sigma}q\preceq_{\sigma}x$) it holds that $d_X(z,q)\le \rho\cdot d_X(x,y)$ (\Cref{def:TriangleLSO}). 
Note that every $(\tau,\rho)$-triangle LSO is also a $(\tau,\rho)$-LSO; however, a  $(\tau,\rho)$-LSO is only a $(\tau,2\rho+1)$-triangle LSO (by the triangle inequality). Hence  a triangle-LSO is preferable to the classic LSO.
We observe that ultrametrics admit a $(1,1)$-triangle-LSO (\Cref{thm:CoverToTriangleLSO}), and show that general $n$-point metric spaces admit an $\left(\tilde{O}(n^{\frac{1}{k}}\cdot\eps^{-1}),2k+\eps\right)$-triangle-LSO (\Cref{cor:TriangleLSOGeneralGraphs}).

We then prove a meta-theorem stating that every metric space admitting a $(\tau,\rho)$-triangle LSO has an oblivious $\nu$-reliable  $2\rho$-spanner with $n\tau\cdot O\left(\log^{2}n+\nu^{-1}\tau\log n\cdot\log\log n\right)$ edges (\Cref{Thm:TriangleLSOtoSpanner}). This gives oblivious reliable spanners for ultrametrics and general metric spaces.
Our spanners for general metrics have significantly smaller stretch compared to \cite{HMO21} ($8k$ compared to $512k$ $^{\ref{foot:MN07}}$), this constant is highly important as it governs the parameter in the power of $n$. An additional advantage is that we remove the dependency on the aspect ratio (which a priori can be unbounded).

\paragraph{Left-sided LSO.~}
A $(\tau,\rho)$-left-sided LSO for a metric space $(X,d_X)$ is a collection $\Sigma$ of linear orderings over \emph{subsets} of $X$, called \emph{partial orderings}, such that every point $x$ belongs to at most $\tau$ partial orderings, and for every $x,y\in X$, there is a partial ordering $\sigma\in\Sigma$ such that for every two points $x',y'\in X$ satisfying  $x\preceq_{\sigma}x'$ and $y'\preceq_{\sigma}y$, it holds that $d_X(x',y')\le \rho\cdot d_X(x,y)$ (\Cref{def:leftsidedLSO}). Note that the \emph{stretch guarantee} of  a $(\tau,\rho)$-left-sided LSO implies that of a $(\tau,\rho)$-LSO (but not the vice versa). However, there could be $\Omega(n)$ (partial) orderings in a $(\tau,\rho)$-left-sided LSO. By lifting the restriction on the total number of partial orderings, we can construct a  left-sided LSO with an optimal stretch of $1$ or a nearly optimal stretch of $1+\eps$; see \Cref{tab:LSO}. This small stretch ultimately leads to the (nearly) optimal stretch for the reliable spanners of tree and planar metrics constructed in this work, which is not attainable in previous work~\cite{HMO21}.

We then prove a meta-theorem stating that every metric space admitting a $(\tau,\rho)$-left-sided LSO has an oblivious $\nu$-reliable $2\rho$-spanner with $n\cdot O(\nu^{-1}\tau^2\log n)$ edges (\Cref{thm:MetaLeftLSOMain}).
We show that $n$-vertex trees admit a $(\log n,1)$-left-sided LSO and conclude that trees have oblivious $\nu$-reliable $2$-spanners with $n\cdot O(\nu^{-1}\log^{3}n)$ edges (\Cref{thm:SpannerForTw}). 
Note that the stretch parameter $2$ is optimal.
Later, we show that planar graphs admit  a $(\frac{1}{\eps}\log^2 n,1+\eps)$-left-sided LSO for every $\eps\in(0,1)$ (\Cref{thm:LeftLSOMinorFree}). An oblivious $\nu$-reliable $(2+\eps)$-spanner with $n\cdot O(\nu^{-1}\eps^{-2}\log^{5}n)$ edges follows (\Cref{thm:SpannerForMinorFree}). The same results also hold for bounded treewidth graphs and graphs excluding a fixed minor.

\paragraph{Ultrametric cover.~} 
A $(\tau,\rho)$-tree cover for a metric space $(X,d_X)$ is a set $\mathcal{T}$ of $\tau$ dominating trees $^{\ref{foot:dominating}}$ such that the distance between every pair of points  is preserved up to a factor $\rho$ in at least one tree ($\forall u,v,~\min_{T\in\mathcal{T}}d_T(u,v)\le \rho\cdot d_X(u,v)$). When all trees in the cover are ultrametrics, we call it an ultrametric cover (\Cref{def:UltrametricCover}).
The first study on tree covers was for Euclidean spaces by Arya \etal \cite{ADMSS95} who constructed the so-called Dumbbell trees. 
For general metrics, Mendel and Naor \cite{MN07} (implicitly) constructed an ultrametric cover from Ramsey type embeddings. These covers actually have a stronger guarantee, where every vertex $v$ is guaranteed to have an ultrametric in the cover approximating its shortest path tree ($\forall v\exists T\forall u,~d_T(v,u)\le \rho\cdot d_X(v,u)$).
There is a long line of work on Ramsey-type embeddings \cite{BFM86,BBM06,BLMN05,MN07,NT12,BGS16,ACEFN20,FL21,Bar21,Fil21}. The state of the art covers follow from Naor and Tao \cite{NT12}, and implies a $(2e\cdot k,O(k\cdot n^{1/k}))$ ultrametric cover.
For doubling metrics, Bartal, Fandina, and Neiman \cite{BFN19Ramsey} constructed a  $(1+\eps,\eps^{-O(d)})$ tree cover.
We refer to \cite{BFN19Ramsey} for further results and background on tree covers.

We observe that every ultrametric admits a $(1,1)$-triangle LSO, which implies that given a $(\tau,\rho)$-ultrametric cover, one can construct a $(\tau,\rho)$-triangle LSO (\Cref{thm:CoverToTriangleLSO}).
Indeed, the main step in our construction of a triangle LSO for general metrics is a construction of a $\left(\tilde{O}(n^{\frac{1}{k}}\cdot\eps^{-1}),2k+\eps\right)$-ultrametric cover (\Cref{thm:GeneralUltrametricCover}). Our construction provides a constant improvement in the stretch parameter (equivalently, a polynomial improvement in the number of ultrametrics in the cover) compared to previous results.

A more structured case is that of a $(\tau,\rho,k,\delta)$-ultrametric cover, where in addition to being a $(\tau,\rho)$-ultrametric cover, we require that each ultrametric will be a $k$-HST of degree at most $\delta$ (see \Cref{def:HST}).
We show that every $\Omega(\frac1\eps)$-HST of degree bounded by $\delta$ admits a (classic) $(\frac\delta2,\eps)$-LSO  (\Cref{lm:UltrametricOrdering}). It follows that  a $(\tau,\rho,\Omega(\frac1\eps),\delta)$-ultrametric cover implies a  $\left(\tau\cdot\frac{\delta}{2}, (1+\eps)\rho\right)$-LSO (\Cref{lem:CoverToLSOClassic}).
The trees in the tree cover for doubling metrics of \cite{BFN19Ramsey} are far from being ultrametrics and cannot be used in our framework.
We then construct an $(\epsilon^{-O(d)},1+\eps,\frac{1}{\epsilon}, \epsilon^{-O(d)})$-ultrametric cover for spaces with doubling dimension $d$ (\Cref{thm:DoublingUltrametricCoverMain}), which implies the respective LSO. Interestingly, having such an ultrametric cover is a characterizing property for metric spaces of bounded doubling dimension  (\Cref{thm:DoublingUltrametricCoverMain}).
See \Cref{tab:TreeCover} for a summary.

\begin{table}[]
	\begin{tabular}{|c|l|l|l|l|}
		\hline
		\multicolumn{1}{|l|}{\textbf{Space}}                                              & \textbf{type} & \textbf{stretch}     & \textbf{\# of trees}                                              & \textbf{ref}                                             \\ \hline
		\multicolumn{1}{|l|}{Euclidean $\mathbb{R}^d$}                               & tree          & $1+\eps$             & $O((\frac{d}{\eps})^d\log\frac{d}{\eps})$                         & \cite{ADMSS95}                          \\ \hline
		\multirow{3}{*}{\begin{tabular}[c]{@{}c@{}}Doubling\\ dimension $d$\end{tabular}} & ultrametric   & $O(d^2)$             & $O(d\log d)$                                                      & \cite{CGMZ16}                           \\ \cline{2-5} 
		& tree          & $1+\eps$             & $\eps^{-O(d)}$                                                    & \cite{BFN19}                            \\ \cline{2-5} 
		& ultrametric   & $1+\eps$             & $\eps^{-O(d)}$                                                    & \Cref{thm:DoublingUltrametricCoverMain} \\ \hline
		\multirow{2}{*}{\begin{tabular}[c]{@{}c@{}}General\\ metric\end{tabular}}         & ultrametric   & $2e\cdot k$          & $O(k\cdot n^{1/k})$                                            & \cite{MN07,NT12}                        \\ \cline{2-5} 
		& ultrametric   & $2k+\eps$ & $\tilde{O}(n^{\frac{1}{k}}\cdot\eps^{-1})$ & \Cref{thm:GeneralUltrametricCover}      \\ \hline
	\end{tabular}
	\caption{\small{New and previous constructions of tree and ultrametric covers.}
	}
	\label{tab:TreeCover}
\end{table}

\paragraph{$2$-hop reliable spanners for the path graph.~}
Using (different types of) LSO, we can reduce the problem of constructing reliable spanners for different complicated metric spaces to that of constructing reliable spanners for the $1$-dimensional path graph. 
Buchin \etal \cite{BHO20} constructed an oblivious $\nu$-reliable $1$-spanner with $n\cdot\tilde{O}\left(\nu^{-1}(\log\log n)^{2}\right)$ edges for the path graph. However, the shortest path between two given vertices in their spanner could contain $\Omega(\log n)$ edges (called \emph{hops}). 
While $(\log n)$-hop spanners are acceptable when applying them upon a $(\tau,\eps)$-LSO, using a $h$-hop spanner of the path graph for $(\tau,\rho)$-triangle-LSO will result in distortion $h\cdot \rho$. It is therefore desirable to minimize the number of hops used by the spanner. Having a $1$-hop spanner will require $\Omega(n^2)$ edges; we thus settle for the next best thing: a $2$-hop reliable spanner.
Specifically,  we construct an oblivious $\nu$-reliable, $2$-hop $1$-spanner with $n\cdot O\left(\log^{2}n+\nu^{-1}\log n\cdot\log\log n\right)$ edges for the path graph (\Cref{lem:2-hopSpanner}).
This spanner is later used in our \mthmref{Thm:TriangleLSOtoSpanner}  to construct reliable spanners for metric spaces admitting a $(\tau,\rho)$-triangle-LSO.

For the left-sided LSO case, we also need a $2$-hop spanner for the path graph. However, the shortest $2$-hop path between  the $i$'th and the $j$'th vertices for $i < j$ must go through a vertex to the left of $i$, i.e., a vertex in $[1,i]$ (as opposed to a vertex in $[i,j]$ in the triangle-LSO case). This requirement inspires us to  define a \emph{left-spanner} (\Cref{def:RelLeftSpanner}); we then construct an oblivious $\nu$-reliable  $2$-hop left-spanner with $n\cdot O(\nu^{-1}\log n)$ edges (\Cref{lem:2HopLeft}). The left-spanner is later used in \mthmref{thm:MetaLeftLSOMain} to construct a reliable spanner from a left-sided LSO.

\begin{table}[]
	\begin{tabular}{|l|l|l|l|l|}
		\hline
		Type                                      & guarantee     & size                                                                 & hops        & ref                                      \\ \hline
		\multirow{3}{*}{$1$ spanner} & Deterministic & $n\cdot O(\log n\cdot \nu^{-6})$                               & $O(\log n)$ & \cite{BHO19}            \\ \cline{2-5} 
		& Oblivious     & $n\cdot O(\nu^{-1}\cdot\log\nu^{-1})$                    & $O(\log n)$ & \cite{BHO20}            \\ \cline{2-5} 
		& Oblivious     & $n\cdot O\left(\log^{2}n+\nu^{-1}\log n\cdot\log\log n\right)$ & $2$         & \Cref{lem:2-hopSpanner} \\ \hline
		Left spanner                & Oblivious     & $n\cdot O(\nu^{-1}\log n)$                                     & $2$         & \Cref{lem:2HopLeft}     \\ \hline
	\end{tabular}
	\caption{\small{Construction of reliable spanners for the line.
			\cite{BHO19} and \cite{BHO20} constructed sparse (both deterministic and oblivious) reliable $1$-spanners for points on the line.  
			However, their spanners have $O(\log n)$ hops, which will incur distortion $O(\rho\cdot\log n)$  when applied on a $(\tau,\rho)$-triangle LSO (with $\rho>1$). We construct a $1$-spanner with only $2$-hops, which we later use to construct a reliable spanner from a triangle-LSO. In addition, we construct a $2$-hop left spanner for the line, which is later used to construct a reliable spanner from a left-sided LSO.}}
	\label{tab:LineSpanners}
\end{table}

\paragraph{Connectivity preservers.~}  While research on reliable spanners for metric spaces has been fruitful, nothing is known for reliable spanners of \emph{graphs}, where we require the spanner to be a subgraph of the input graph.  In a recent talk, Har-Peled~\cite{Har-Peled20} asked a ``probably much harder question'': whether it is possible to construct a non-trivial subgraph reliable spanner.  We show that, even for a much simpler problem where one seeks a subgraph to  only preserve \emph{connectivity} for vertices outside $B^+$, the faulty extension of $B$,  the subgraph must have $\Omega(n^2)$ edges in the worst case. Indeed, our lower bound is much more general: it applies to \emph{$g$-reliable} connectivity preservers for some function $g$. A $g$-reliable connectivity preserver of a graph $G=(V,E)$ is a subgraph  $H$ of $G$ such that for every attack $B\subseteq V$, there is a superset $B^+\supseteq B$
of size at most $g(|B|)$, such that for every $u,v\in V\setminus B^{+}$, if $u$ and $v$ are connected in $G\setminus B$, then they are also connected in $H\setminus B$. Observe that a $\nu$-reliable spanner defined in \Cref{def:reliableSpanner} is a $g$-reliable (non-subgraph) spanners for the linear function $g(x)=(1+\nu)x$. We showed that there is an $n$-vertex graph $G$ such that every oblivious $g_{k}$-reliable connectivity preserver has $\Omega(n^{1+1/k})$	edges for any function $g = O(x^k)$ (see \Cref{thm:preserverLowerbound}). Taking $k = 1$ gives a lower bound $\Omega(n^2)$ on the number of edges of \emph{subgraph} $\nu$-reliable spanners.  
On the positive side, we provide a construction of a  \emph{deterministic} connectivity preserver matching the lower bound (\Cref{thm:preserverUB}).

\subsection{Related work}

The tradeoff between stretch and sparsity (number of edges) of (regular) $t$-spanners has been extensively studied \cite{Yao82,Clarkson87,Keil88,PS89,ADDJS93,CDNS95,HIS13,CGMZ16,FN18,LS19}; see the recent survey of Ahmed \etal \cite{ABSHJKS20}, and the book \cite{NS07} and references therein for  more details. The bottom line is that $n$-point metric spaces admit $(2k-1)$-spanners (for every integer $k$) with $O(n^{1+1/k})$ edges \cite{CDNS95}, while the metric induced by $n$ points in $d$ dimensional Euclidean space admits  a $(1+\eps)$-spanner with $n\cdot O(\eps)^{1-d}$ edges \cite{ADDJS93}. Similarly, $n$-point metric spaces with doubling dimension $d$  admit  $(1+\eps)$-spanners with $n\cdot \eps^{-O(d)}$ edges \cite{CGMZ16}.

For vertex-fault-tolerant spanner, it was shown that every $n$-point set in $\mathbb{R}^d$, or more generally in a space of doubling dimension $d$, admits an $f$-\emph{vertex-fault-tolerant} $(1+\eps)$-spanner with $\eps^{-O(d)}\cdot f\cdot n$ edges \cite{LNS02,Lukovszki99,Solomon14}.
For general graphs, after a long line of works  \cite{CLPR10,DK11,BDPW18,BP19Spanners,DR20,BDR21}, it was shown that every $n$-vertex graph admits an efficiently constructible $f$-vertex-fault-tolerant $(2k-1)$-spanner with $O(f^{1-1/k}\cdot n^{1+1/k})$ edges, which is optimal assuming the Erd\"{o}s' Girth Conjecture~\cite{Erdos64}.
A related notion is that of a vertex-fault-tolerant (VFT) emulator. Unlike spanners, emulators are not required to be subgraphs, and the weight of an emulator edge is determined w.r.t. the faulty set. 
It was recently shown that vertex-fault-tolerant (VFT) emulators are asymptotically sparser from their spanner counterparts \cite{BDN22}.

In addition to vertex-fault-tolerant (VFT) spanners, also  edge-fault-tolerant (EFT) spanners were studied, where the guarantee is to withstand up to $f$-edge faults (as opposed to $f$ vertex faults in VFT). 
Bodwin, Dinitz, and Robelle \cite{BDR22} constructed an $f$-EFT $2k-1$-spanners with $O(k^2f^{\frac12-\frac{1}{2k}}\cdot n^{1+\frac1k}+kfn)$~$\big/$~$O(k^2f^{\frac12}\cdot n^{1+\frac1k}+kfn)$ edges for odd$\big/$even values of $k$ respectively. 
There is also a lower bound of  $\Omega(f^{\frac12-\frac{1}{2k}}n^{1+1/k})$ \cite{BDPW18}.

Abam \etal\cite{ABFG09} introduced the notion of \emph{region} fault-tolerant spanners for the Euclidean plane. They showed that one can construct a $t$-spanner with $O(n\log n)$ edges in such a way that if points belonging to a convex region are deleted, the residual graph is still  a spanner for the remaining points. 

Spanners with low hop diameter  for Euclidean spaces of fixed dimension  were studied in the pioneering work of Arya \etal\cite{AMS94}. State of the art is a $(1+\epsilon)$-spanner constructible in $O(n\log n)$ time by Solomon~\cite{Solomon13} that has $O(n\alpha_k(n))$ \footnote{$\alpha_k(n)$ is the inverse function of  a very fast growing function at level $k$ of the the primitive recursive hierarchy; see~\cite{Solomon13} for a more formal description.} edges and hop diameter $k$.

In addition to having a small number of edges, it is desirable to have a spanner with a small total edge weight, called a \emph{light spanner}. Light spanners have been thoroughly studied in the spanner literature \cite{CDNS95,ES13,ENS15,Gottlieb15,CW18,FS20,BLW17,BLW19,ADFSW19,LS19,CFKL20}.
Sparse (and light) spanners were constructed efficiently in different computational models such as  LOCAL \cite{KPX08}, CONGEST \cite{EFN20}, streaming \cite{Elkin11,Baswana08,BS07,KW14,AGM12Spanners,FKN20}, massive parallel computation (MPC) \cite{BDGMN20} and dynamic graph  algorithms~\cite{BKS12, BFH19}.

\input{Prelim} 
\input{ultrametric}

\input{LSO}
\input{triangleLSO}

\input{LeftSidedLSO}

\input{Subgraph}

 \section{Conclusions}\label{sec:conclusion}
 In this paper, we have presented different types of locality-sensitive orderings and used them to construct reliable spanners.
 For the construction of the LSO's, we introduced and constructed ultrametric covers. Finally, in order to use the LSO's to construct reliable spanners, we construct $2$-hop spanners and left spanners for the path graph. Several open questions naturally arise from our work: 
 \begin{enumerate}
 	\item Can we construct a $\nu$-reliable $2$-hop $1$-spanner for the path graph $P_n$ with $O(n\log n)$ edges for constant $\nu$?  Note that a $2$-hop spanner for the path graph $P_n$ even without any reliability guarantee must contain $\Omega(n\log n)$ edges (see Excercise 12.10~\cite{NS07}).
 	\item A major open question is the construction of deterministic reliable spanners general metric spaces. \cite{HMO21} constructed deterministic reliable $O(t^2)$-spanners for general metrics with $\tilde{O}(n^{1+\frac1t})$ edges, and deterministic reliable $O(t)$-spanners for trees and planar graphs with $\tilde{O}(n^{1+\frac1t})$ edges, while showing an $\Omega(n^{1+\frac1t})$ lower bound on the number of edges in a deterministic reliable $t$-spanners for the uniform metric.
 	Using the new LSO's constructed in this paper, it could be possible to improve the stretch parameters by a constant factor and remove the dependency on aspect ratio from the sparsity. However, the lower bound for uniform metrics applies to trees and planar graphs as well; thus using $\tilde{O}(n^{1+\frac1t})$ edges to obtain stretch $t$ is necessary. 
 	On the other hand, general metrics are far from understood. Closing the gap between the current $O(t^2)$ upper bound to the $t$ lower bound is a fascinating open question. 
  	\item Can we construct a reliable spanner of stretch $2$ (as opposed to stretch $2+\epsilon$ presented in this paper) for planar metrics with a nearly linear number of edges?
 \end{enumerate}

 \bibliographystyle{alphaurlinit}
 \bibliography{RelaibleSpannerBib,RPTALGbib}
 \appendix   

\end{document}

%% file: Prelim.tex
\section{Preliminaries}\label{sec:prelim}

Let $(X,d_X)$ be a metric space.  The aspect ratio, or spread, denoted by $\Phi$, is defined as follows: $\Phi = \frac{\max_{x,y\in X}d_X(x,y)}{\min_{x\not=y \in X}d_X(x,y)}$.
We denote by $[n]$ the set of integers $\{1,2,\ldots, n\}$. For two integers $a\leq b$, we define $[a:b] = \{a,a+1,\ldots, b\}$. 

We use $\tilde{O}$ notation hides poly-logarithmic factors. That is $\tilde{O}(f)=O(f)\cdot\log^{O(1)}(f)$.

Let $G$ be a graph. We denote the vertex set and edge set of $G$ by $V(G)$ and $E(G)$, respectively. When we want to explicitly specify the vertex set $V$ and edge set $E$ of $G$, we write $G=(V,E)$. If $G$ is a weighted graph, we write $G=(V,E,w)$ with $w:E\rightarrow\R_+$ being the weight function on the edges of $G$. 
For every pair of vertices $x,y\in V$, we denote by $d_G(x,y)$ the shortest path distance between $x$ and $y$  in $G=(V,E,w)$. Given a path $P\subseteq G$, we define the \emph{hop length} of $P$ to be the number of edges on the path. 

A \emph{$t$-spanner} for a metric space $(X,d_X)$ is a weighted graph $H(V,E,w)$ that has $V = X$, $w(u,v) = d_X(u,v)$ for every edge $(u,v)\in E$ and $d_X(x,y) \leq d_H(x,y) \leq t\cdot d_X(x,y)$ for every pair of points $x,y\in X$. We say that a $t$-spanner $H$ has hop number $h$ if for every pair of vertices $x,y$, there is an $x-y$ path $P$ in $H$ of at most $h$ hops such that $w_H(P)\le t\cdot d_X(x,y)$.

The path graph $P_{n}$ contains $n$ vertices
$v_{1},v_{2},\dots,v_{n}$ and there is (unweighted) edge between
$v_{i}$ and $v_{j}$ iff $|i-j|=1$. A path $v_{i_{1}},v_{i_{2}},\dots v_{i_{s}}$
is monotone iff for every $j$, $i_{j}<i_{j+1}.$ Note that if a spanner
$H$ contains a monotone path between $v_{i},v_{j}$ then $d_{H}(v_{i},v_{j})=d_{P_{n}}(v_{i},v_{j})=|i-j|$. We sometimes identify vertices of $P_n$ with numbers in $\{1,2,\ldots, n\}$, and refer to $\{1,2,\ldots, n\}$ as the vertex set of $P_n$.

A metric $(X,d_X)$ has doubling dimension $d$ if every ball of radius $r$ can be covered by at most $2^{d}$ balls of radius $r/2$. The following lemma gives the standard packing property of doubling metrics (see, e.g., \cite{GKL03}).
\begin{lemma}[Packing Property]\label{lem:doubling_packing}
	Let $(X,d)$ be a metric space  with doubling dimension $d$.
	If $S \subseteq X$ is a subset of points with minimum interpoint distance $r$ that is contained in a ball of radius $R$, then 
	$|S| = \left(\frac{2R}{r}\right)^{O(d)}$ .
\end{lemma}

In the following lemma, we show that when constructing oblivious spanners, it is enough to bound the number of edges in expectation to obtain a worst-case guarantee.

\begin{lemma}\label{lem:ExpectedSize}
	Consider an $n$-vertex graph $G=(V,E,w)$ that admits an oblivious $\nu$-reliable $t$-spanner with $m$ edges in expectation.
	Then $G$ admits an oblivious $2\nu$-reliable $t$-spanner with $2m$ edges in the worst case.
\end{lemma}
\begin{proof}
	Formally, there is a distribution $\mathcal{D}$ over spanners $H$ such that for every attack $B\subseteq V$, $\mathbb{E}[|B^+\setminus B|]\le \nu|B|$, and $\mathbb{E}[|H|]\le m$.
	Let $\mathcal{D}'$ be the distribution over spanners $H$ obtained by conditioning $\mathcal{D}$ on the event $|H|\le 2m$. Clearly, all the spanners in $\supp\{\mathcal{D}'\}$ have at most $2m$ edges. Furthermore, for every attack $B\subseteq V$, it holds that
	\begin{align*}
	\mathbb{E}_{H\sim\mathcal{D}'}[|B^{+}\setminus B|] & =\mathbb{E}_{H\sim\mathcal{D}}[|B^{+}\setminus B|~\bigl|~|H|\le2m]\\
	& =\frac{1}{\Pr\left[|H|\le2m\right]}\cdot\left(\mathbb{E}_{H\sim\mathcal{D}'}[|B^{+}\setminus B|~\bigl| ~|H|\le2m]\cdot\Pr\left[|H|\le2m\right]\right)\\
	& \le\frac{1}{\Pr\left[|H|\le2m\right]}\cdot\mathbb{E}_{H\sim\mathcal{D}}[|B^{+}\setminus B|]\le2\nu\cdot|B|~,
	\end{align*}
	where in the last inequality, we use Markov's inequality.
\end{proof}

%% file: ultrametric.tex
\section{Ultrametric Covers}\label{sec:ultrametricCover}

\paragraph{Ultrametric.~} An ultrametric $\left(X,d\right)$ is a metric space satisfying a strong form of the triangle inequality, that is, for all $x,y,z\in X$,
$d(x,z)\le\max\left\{ d(x,y),d(y,z)\right\}$. A related notion is a $k$-hierarchical well-separated tree ($k$-HST).

\begin{definition}[$k$-HST]\label{def:HST}
	A metric $(X,d_X)$ is a $k$-hierarchical well-separated tree ($k$-HST) if there exists a bijection $\varphi$ from $X$ to leaves of a rooted tree $T$ in which:
	\begin{enumerate}[noitemsep]
		\item Each node $v\in T$ is associated with a label $\Gamma_{v}$ such that $\Gamma_{v} = 0$ if $v$ is a leaf and $\Gamma_{v}\geq k\Gamma_{u}$ if $v$ is an internal node and $u$ is any child of $v$.
		\item $d_X(x,y) = \Gamma_{\lca(\varphi(x),\varphi(y))}$ where $\lca(u,v)$ is the least common ancestor of any two given nodes $u,v$ in $T$. 
	\end{enumerate}
\end{definition}
 
 It is well known that any ultrametric is a $1$-HST, and any $k$-HST is an ultrametric (see \cite{BLMN05}).

\paragraph{Ultrametric cover.~}  Consider a  metric space $(X,d_X)$, a distance measure $d_Y$ is said to be dominating if $\forall x,y\in X$, $d_X(x,y)\le d_Y(x,y)$. A tree/ultrametric over $X$ is said to be dominating if their metric is dominating.
Bartal, Fandina, and Neiman \cite{BFN19Ramsey} studied \emph{tree covers}:  a  metric space $(X,d_X)$ admits a \emph{$(\tau,\rho)$-tree cover} if there are at most $\tau$ dominating trees $\{T_1,T_2,\dots,T_{\tau}\}$ such that $X\subseteq V(T_i)$ for every $i\in [\tau]$ and for every pair of points $x,y \in X$, there is some tree $T_i$ where $d_{T_i}(x,y)\le\rho\cdot d_X(x,y)$. Bartal \etal \cite{BFN19Ramsey} observed that the previous constructions of Ramsey trees\footnote{Ramsey trees have additional de`sired property compared to general tree covers: for every vertex $v$, there is a single tree in the cover satisfying all its pairwise distances, as oppose to union of all the trees in a general tree cover.} \cite{MN07,NT12,ACEFN20} give an  $(\tilde{O}(n^{1/k}),2ek)$-tree cover for general metrics, and explicitly constructed an  $(\eps^{-O(d)},1+\eps)$-tree cover for metric spaces with doubling dimension $d$.
Here we initiate the study of ultrametric covers.
\begin{definition}[Ultrametric Cover]\label{def:UltrametricCover}
	A \emph{$(\tau,\rho)$-ultrametric cover} for a space $(X,d)$ is a collection of at most $\tau$ dominating ultrametrics $\mathcal{U} = \{(U_i,d_{U_i})\}_{i=1}^{\tau}$ over $X$, such that for every $x,y\in X$ there is an ultrametric $U_i$ for which $d_{U_i}(x,y)\le \rho\cdot d_X(x,y)$.
	
	If every metric $(U,d_{U}) \in \mathcal{U}$ is a $k$-HST and the corresponding tree $T_{U}$ of $U$ has maximum degree at most $\delta$, we say that $\mathcal{U}$ is a  \emph{$(\tau,\rho,k,\delta)$-ultrametric cover} of $(X,d_X)$.
\end{definition} 
Note that ultrametrics are much more structured than general trees. For example, every ultrametric embeds isometrically into $\ell_{2}$, while trees require distortion $\sqrt{\log\log n}$ \cite{Bourgain86,Mat99}.
Later, we will show how to use ultrametric covers to construct locality-sensitive orderings (see \Cref{lem:CoverToLSOClassic} and \Cref{thm:CoverToTriangleLSO}).

The first main result of this section is \Cref{thm:GeneralUltrametricCover} where we construct an ultrametric cover for general metrics. 
\begin{restatable}[Ultrametric Cover For General Metrics]{theorem}{GeneralUltrametricCover}
	\label{thm:GeneralUltrametricCover}
	For every $k\in \mathbb{N}$ and $\eps\in(0,\frac12)$, every $n$-point metric space admits an   $\left(O(n^{\frac{1}{k}}\cdot\log n\cdot\frac{k^{2}}{\eps}\cdot\log\frac{k}{\eps}),2k+\eps\right)$-ultrametric cover.
\end{restatable}

Interestingly, the tree cover in \cite{BFN19Ramsey} for general metrics actually consists of ultrametrics; in other words, Bartal \etal \cite{BFN19Ramsey} obsereved that Ramsey trees constitute an $(\tilde{O}(n^{1/k}),2ek)$-ultrametric cover. Thus, we obtain a polynomial  improvement in the number of ultrametrics in the cover. Specifically, to guarantee stretch $\approx2(k+1)$, our cover uses  $\tilde{O}(n^{1/k})$ ultrametrics, while previous covers have  $\Omega(n^{\frac{e}{k+1}})$ ultrametrics.

Next, in \Cref{thm:DoublingUltrametricCoverMain} below, we show that every metric space with doubling dimension $d$ admits an $(\epsilon^{-O(d)},1+\eps,\frac{1}{\epsilon}, \epsilon^{-O(d)})$-ultrametric cover for any parameter $\eps \in (0,\frac{1}{6})$. It turns out that this property is actually a characterization of doubling spaces. The proof of~\Cref{thm:DoublingUltrametricCoverMain} is provided in~\Cref{subsec:UltraCoverDoubling}.

\begin{restatable}[Ultrametric Cover For Doubling Metrics]{theorem}{DoublingUltrametricCover}
	\label{thm:DoublingUltrametricCoverMain}
	Every metric space $(X,d_{X})$ with doubling dimension $d$ admits an $(\epsilon^{-O(d)},1+\eps,\frac{1}{\epsilon}, \epsilon^{-O(d)})$-ultrametric cover for any parameter $\eps \in (0,\frac{1}{6})$. 
	
	Conversely, if a metric space $(X,d_X)$ admits a $(\tau, \rho, k, \delta)$-ultrametric cover for $k \geq 2\rho$, then it has doubling dimension $d\le\log(\tau\delta)$.
\end{restatable}

The main tool in proving \Cref{thm:GeneralUltrametricCover,thm:DoublingUltrametricCoverMain} is \emph{pairwise partition cover}, a newly introduced notion, which is closely related to the previously introduced stochastic/padded decompositions and sparse covers \cite{AP90,KPR93,Bar96,GKL03,AGGNT19,Fil19padded}.
A partition $\mathcal{P}$ of a metric space $(X,d_X)$ is $\Delta$-bounded if every cluster $C\in\mathcal{P}$ has diameter at most $\Delta$.
\begin{definition}[Pairwise Partition Cover Scheme]
	A collection of partitions $\mathbb{P} = \{\mathcal{P}_{1},\dots,\mathcal{P}_{s}\}$
	is $(\tau,\rho,\eps,\Delta)$-pairwise partition cover if (a) $s\le \tau$, (b) every partition $\mathcal{P}_{i}$ is $\Delta$-bounded,
	and (c) for every pair $x,y$ such that $\frac{\Delta}{2\rho}\le d_X(x,y)\le\frac{\Delta}{\rho}$, there is a cluster $C$ in one of the partitions $\mathcal{P}_{i}$ such that $C$ contains both closed balls $B(x,\eps\Delta),B(y,\eps\Delta)$.\\
	A space $(X,d_{X})$ admits a $(\tau,\rho,\eps)$-\emph{pairwise partition cover scheme} if for every $\Delta$, it admits a  $(\tau,\rho,\eps,\Delta)$-pairwise partition cover.
\end{definition} 

We will show that given a pairwise partition cover scheme, one can construct an ultrametric cover. The proof appears in \Cref{sec:ProofOfPairwsieCoverReduction}.
\begin{restatable}[]{lemma}{PairwiseCoverToUltrametricCover}
	\label{lem:PairwisePartitionCoverSchemeToUltrametric}
	Suppose that a metric space $(X,d_{X})$ admits a $(\tau,\rho,\epsilon)$-pairwise partition cover scheme for $\tau\in\N$, $\rho\ge1$, and $\eps\in(0,\frac12)$. Then $X$  admits an  $\left(O(\frac{\tau}{\epsilon}\log\frac{\rho}{\epsilon}),\rho(1+7\epsilon)\right)$-ultrametric cover.\\
	Furthermore, every ultrametric in the cover is a $\Theta(\frac{\rho}{\epsilon})$-HST.
\end{restatable}

In \Cref{sec:PairwiseCoverGeneralMetrics} we construct a pairwise partition cover for general metrics:
\begin{restatable}[]{lemma}{pairwisePartitionGeneral}
	\label{lem:pairwise-partition-general}
	Every $n$-point metric space $(X,d_X)$ admits an 
	$(O(n^{\frac{1}{k}}\log n), 2k+\delta,\frac{\delta}{2k(2k+\delta)})$-pairwise partition cover scheme for any $\delta \in [0,1]$ and integer $k\ge 1$.
\end{restatable}

We are now ready to prove \Cref{thm:GeneralUltrametricCover}:
\begin{proof}[Proof of \Cref{thm:GeneralUltrametricCover}]
	Let $(X,d_X)$ be an $n$-point metric space, and fix $\delta=\frac\eps8$.
	By \Cref{lem:pairwise-partition-general}, 	
	$X$ admits an $(O(n^{\frac{1}{k}}\log n), 2k+\delta,\frac{\delta}{2k(2k+\delta)})$-pairwise partition cover.
	By \Cref{lem:PairwisePartitionCoverSchemeToUltrametric},
	$X$ admits an ultrametric cover with $O(\frac{n^{\frac{1}{k}}\log n}{\delta}\cdot2k(2k+\delta)\cdot\log\frac{(2k+\delta)2k(2k+\delta)}{\delta}=O(n^{\frac{1}{k}}\log n\cdot\frac{k^{2}}{\delta}\cdot\log\frac{k}{\delta})=O(n^{\frac{1}{k}}\log n\cdot\frac{k^{2}}{\eps}\cdot\log\frac{k}{\eps})$
	ultrametrics, and stretch $(2k+\delta)(1+\frac{7\delta}{2k(2k+\delta)})<2k+8\delta=2k+\eps$.	
\end{proof}

\subsection{From Pairwise Partition Cover to Ultrametric Cover: Proof of \Cref{lem:PairwisePartitionCoverSchemeToUltrametric}}\label{sec:ProofOfPairwsieCoverReduction}
\Cref{lem:PairwisePartitionCoverSchemeToUltrametric} is a reduction from pairwise partition cover scheme to ultrametric cover. 
In essence, an ultrametric is simply a hierarchical partition.
Thus, this reduction takes unrelated partitions in all possible scales, and combines them into hierarchical/laminar partitions. Reductions similar in spirit were constructed in the context of the Steiner point removal problem \cite{Fil20scattering}, stochastic Steiner point removal \cite{EGKRTT14}, universal Steiner tree \cite{BDRRS12}, and others.
We follow here a bottom-up approach, where the ratio between consecutive scales in a single hierarchical partition (a.k.a. ultrametric) is $O(\frac\rho\eps)$. When constructing the next level in the  hierarchical partition, we take partitions from a pairwise partition cover of the current scale, and slightly ``round'' them around the ``borders'' so that no previously created cluster will be divided (see \Cref{fig:Laminar}). The argument is that due to the large ratio between consecutive scales, the effects of this rounding are marginal.
\begin{proof}[Proof of \Cref{lem:PairwisePartitionCoverSchemeToUltrametric}.]
	Assume w.l.o.g. that the minimal pairwise distance in $X$ is $1$, while the maximal  pairwise distance is $\Phi$.	
	Fix 
	$c\ge 1$ to be determined later.
	For $i\ge0$, set
	$\Delta_{i}=c\cdot(\frac{4\rho}{\epsilon})^{i}$, and let
	$\mathbb{P}_i=\{\mathcal{P}_{1}^{i},\dots,\mathcal{P}_{\tau}^{i}\}$ be a $(\tau,\rho,\Delta_{i})$-padded partition cover (we assume that $\mathbb{P}_i$ has exactly $\tau$ partitions; we can enforce this assumption by duplicating partitions if necessary). 
	Fix some $j$, let $\mathcal{P}^{-1}_{j}$ be the partition where each vertex is a singleton, and consider $\{\mathcal{P}^{i}_{j}\}_{i\ge -1}$. We will inductively define a new set of partitions, enforcing it to be a laminar system. The basic idea of doing this is to produce a tree of partitions where the lower level is a refinement of the higher level, and we do so by grouping a cluster at a lower level to one of the clusters at a  higher level separating it. 
	
	The lowest level $\mathcal{P}^{-1}_{j}$ where each set in the partition is a singleton, stays as-is. 
	Inductively, for any $i\geq 0$, after constructing $\tilde{\mathcal{P}}_{j}^{i-1}$ from $\mathcal{P}_{j}^{i-1}$, we will construct $\tilde{\mathcal{P}}_{j}^{i}$
	from $\mathcal{P}_{j}^{i}$ using $\tilde{\mathcal{P}}_{j}^{i-1}$.
	Let $\mathcal{P}_{j}^{i}=\left\{ C_{1},\dots,C_{\phi}\right\}$ be the clusters in the partition $\mathcal{P}_{j}^{i}$. For each $q\in[1,\phi]$, let $Y_{q}=X\setminus\cup_{a<q}\tilde{C}_{a}$ be the set of unclustered points (w.r.t. level $i$, before iteration $q$). 
	Let	$C'_{q}=C_{q}\cap Y_{q}$ be the cluster $C_q$ restricted to vertices in $Y_q$, and let $S_{C'_{q}}=\left\{ C\in\tilde{\mathcal{P}}_{j}^{i-1}\mid C\cap C'_{q}\ne\emptyset\right\}$ be  the set of new level-$(i-1)$ clusters with non empty intersection with $C'_{q}$.
	We set the new cluster $\tilde{C}_{q}=\cup S_{C'_{q}}$ to be the union of these clusters, and continue iteratively.
	See \Cref{fig:Laminar} for illustration.
	Clearly, $\tilde{\mathcal{P}}_{j}^{i-1}$
	is a refinement of $\tilde{\mathcal{P}}_{j}^{i}$. We conclude that
	$\left\{ \tilde{\mathcal{P}}_{j}^{i}\right\} _{i\ge-1}$ is a laminar hierarchical set of partitions that refine each other.
	
	\begin{figure}[t]
		\centering
		\includegraphics[width=1\textwidth]{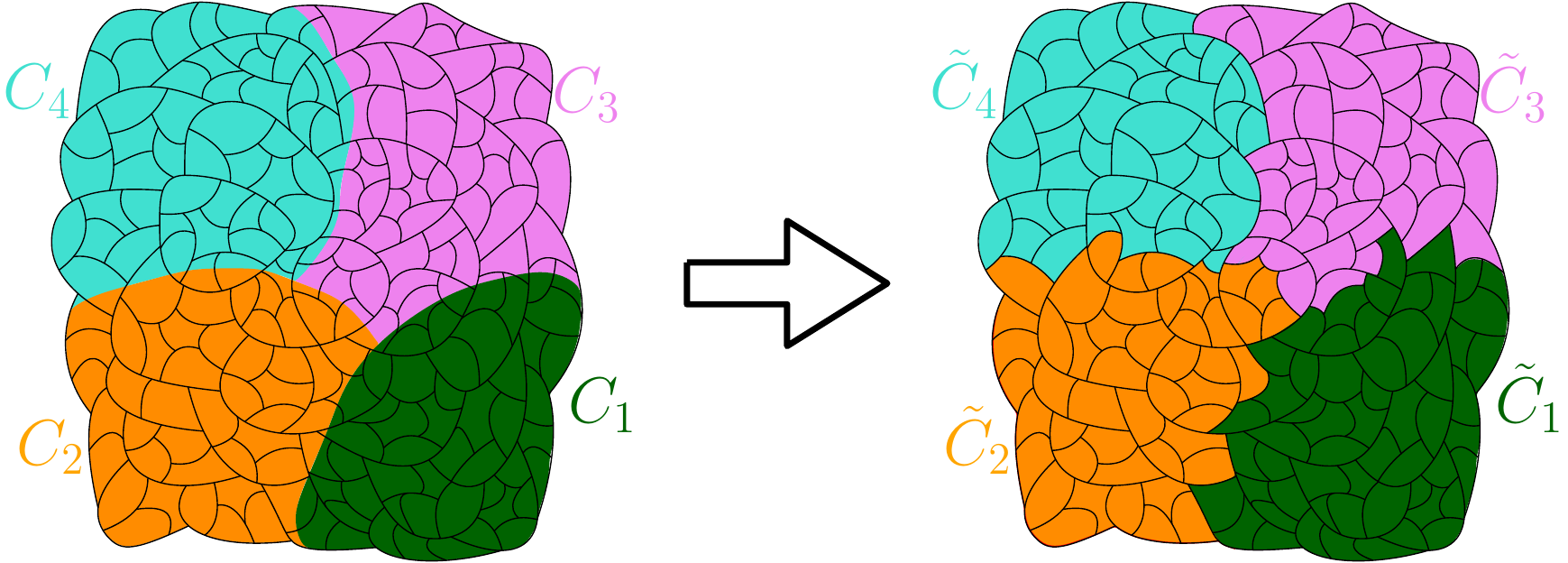}
		\caption{\footnotesize{Illustration of the construction of the partition $\tilde{\mathcal{P}}_{j}^{i}$ given  $\mathcal{P}_{j}^{i}$ and $\tilde{\mathcal{P}}_{j}^{i-1}$.
				The black lines in both the left and right parts of the figure border clusters in $\tilde{\mathcal{P}}_{j}^{i-1}$. On the left illustrated the partition $\mathcal{P}_{j}^{i}=\{C_1,C_2,C_3,C_4\}$, where different clusters colored by different colors. On the right illustrated the modified partition $\mathcal{P}_{j}^{i}=\{\tilde{C}_1,\tilde{C}_2,\tilde{C}_3,\tilde{C}_4\}$. $\tilde{C}_1$ contains all the clusters in $\tilde{\mathcal{P}}_{j}^{i-1}$ intersecting $C_1$. $C_2$ contains all the clusters in $\tilde{\mathcal{P}}_{j}^{i-1}\setminus \tilde{C}_1$ intersecting $C_2$, and so on.
		}}
		\label{fig:Laminar}
	\end{figure}
	
	We next argue by induction that $\tilde{\mathcal{P}}_{j}^{i}$ has diameter $\Delta_{i}(1+\epsilon)$.
	Consider $\tilde{C}_{q}\in\tilde{\mathcal{P}}_{j}^{i}$; it consists
	of $C'_{q}\subseteq C_{q}\in\mathcal{P}_{j}^{i}$ and of clusters in $\tilde{\mathcal{P}}_{j}^{i-1}$
	intersecting $C'_{q}$. 
	As the diameter of $C'_{q}$ is bounded by $\diam(C_{q})\le \Delta_i$, and by the induction hypothesis, the diameter of each cluster $C\in \tilde{\mathcal{P}}_{j}^{i-1}$ is bounded by $(1+\eps)\Delta_{i-1}$, we conclude that the diameter of $\tilde{C}_{q}$ is bounded by
	\[
	\Delta_{i}+2\cdot(1+\epsilon)\Delta_{i-1}=\Delta_{i}\left(1+\frac{2(1+\epsilon)}{4\rho}\cdot\epsilon\right)\le\Delta_{i}(1+\epsilon)~,
	\]
	since $\rho \geq 1$ and $\eps < 1$.
	
	Next we argue that 
	$ \tilde{\mathbb{P}}_i = \{\tilde{\mathcal{P}}_{1}^{i},\dots,\tilde{\mathcal{P}}_{\tau}^{i}\}$ is a $(\tau,(1+\eps)\rho,0,(1+\eps)\Delta)$-pairwise partition cover. Observe that it contains $\tau$ partitions, and we have shown that all the clusters have diameter at most $(1+\eps)\Delta$. Thus, it remains to prove that
	for every pair $x,y$  at distance $d_{X}(u,v)\in\left[\frac{(1+\eps)\Delta_{i}}{2(1+\eps)\rho},\frac{(1+\eps)\Delta_{i}}{(1+\eps)\rho}\right]=\left[\frac{\Delta_{i}}{2\rho},\frac{\Delta_{i}}{\rho}\right]$
	contained in some cluster.
	As $d_{X}(u,v)\in\left[\frac{\Delta_{i}}{2\rho},\frac{\Delta_{i}}{\rho}\right]$, there is some index $j$ such that $B_X(u,\eps\Delta_i),B_X(v,\eps\Delta_i)\subseteq C_{i}\in\mathcal{P}_{j}^{i}$. That is, the balls of radius $\eps\Delta_i$ around $u,v$ are contained in a cluster of $\mathcal{P}_{j}^{i}$.
	We argue that
	$u,v\in \tilde{C}_{i}\in\tilde{\mathcal{P}}_{j}^{i}$.
	Let $\tilde{C}_{v},\tilde{C}_{u}\in\tilde{\mathcal{P}}_{j}^{i-1}$ be the clusters containing $v,u$ respectively at $(i-1)$-th level. Note that they both have diameter at most $(1+\eps)\Delta_{i-1}=\frac{(1+\epsilon)\eps}{4\rho}\Delta_{i}<\eps\Delta_{i}$. Hence $\tilde{C}_{v}\subseteq B_X(v,\eps \Delta_i)\subseteq C_{i}$, and similarly $\tilde{C}_{u}\subseteq C_{i}$.
	By the partitioning algorithm, it follows that $\tilde{C}_{u},\tilde{C}_{v}\subseteq \tilde{C}_{i}$ (as $\tilde{C}_u,\tilde{C}_v$ do not intersect any other clusters), and in particular $u,v\in \tilde{C}_{i}$ as required.
		
	Finally, we construct an ultrametric cover.
	Fix an index $j\in [1,\tau]$; we construct a $(\frac{4\rho}{\eps})$-HST $U_j$ as follows. Leaves of $U_j$ bijectively correspond to points in $X$ and have label $0$. For each $i \in [0,I]$ where $I = \lceil \log_{\nicefrac{4\rho}{\epsilon}}\nicefrac\Phi c\rceil$, internal nodes at level $i$ bijectively correspond to the clusters $\tilde{\mathcal{P}}_j^i$
	(leaves of $U_j$ is at level $-1$), and have label $(1+\eps)\Delta_i$. There is an edge from each node corresponding to a cluster $\tilde{C}_{i-1}\in\tilde{\mathcal{P}}_j^{i-1}$ to the node corresponding the unique cluster  $\tilde{C}_{i}\in\tilde{\mathcal{P}}_j^{i}$ containing $\tilde{C}_{i-1}$.
	The root of $U_j$ is the unique single cluster in $\tilde{\mathcal{P}}_j^{I}$. 
	Clearly, the ultrametric cover $\{U_j\}_{j=1}^{\tau}$ is dominating, and every ultrametric is a $\frac{4\rho}{\eps}$-HST. 
	
	To bound the stretch, we will construct such an ultrametric cover with  $c=(1+\eps)^l$ for every $l \in [0, \lfloor\log_{1+\epsilon}\frac{4\rho}{\epsilon}\rfloor]$. The final ultrametric cover will be a union of these $O(\log_{1+\epsilon}\frac{4\rho}{\epsilon})$ ultrametric covers. Clearly, their cardinality is bounded by $\tau\cdot O(\log_{1+\epsilon}\frac{4\rho}{\epsilon})=O(\frac{\tau}{\eps}\log\frac{\rho}{\epsilon})$.
	
	Consider a pair $x,y\in X$.
	Let $l\in[0,\lfloor\log_{1+\epsilon}\frac{4\rho}{\epsilon}\rfloor]$, and $i\ge0$ be the unique indices such that $(1+\epsilon)^{l-1}(\frac{4\rho}{\epsilon})^i\leq (1+\epsilon)\rho \cdot d_{X}(x,y)\le(1+\epsilon)^{l}(\frac{4\rho}{\epsilon})^i$. 
	For $c=(1+\epsilon)^{l}$, there is some index $j$, and a cluster $\tilde{C}_i\in\tilde{\mathcal{P}}_j^i$ such that $x,y\in\tilde{C}_{i}\in\mathcal{\tilde{P}}_{j}^{i}$. Thus in the corresponding ultrametric, $x,y$ both decedents of an internal node with label
	$(1+\epsilon)^{l+1}(\frac{4\rho}{\epsilon})^{i}\le(1+\eps)^{3}\rho\cdot d_{X}(x,y)$, the stretch guarantee follows.
	
	In summary, we have constructed an $\left(O(\frac{\tau}{\epsilon}\log\frac{\rho}{\epsilon}),\rho(1+7\epsilon)\right)$-ultrametric cover, consisting of $(\frac{4\rho}{\eps})$-HST's. 
\end{proof}

\subsection{Pairwise Partition Cover for General Metrics: Proof of \Cref{lem:pairwise-partition-general}}\label{sec:PairwiseCoverGeneralMetrics}

Fix parameter $\delta\in(0,1]$. We begin by creating a distribution over partitions, such that for every pair of points $u,v$ at distance $\frac{\Delta}{2k+\delta}$, there is a non trivial probability that the some balls around $u,v$ contained in a single cluster.
Later, \Cref{lem:pairwise-partition-general} will follow by taking the union of many independently sampled such partitions.
\begin{lemma}\label{lm:pairwise-partition-general}
	For every $n$-point metric space $(X,d_X)$, integer $k\ge 1$, $\delta\in[0,1]$, and $\Delta>0$ there is a distribution over $\Delta$-bounded partitions such that for every pair of points $u,v$ where $d_X(u,v)\le \frac{\Delta}{2k+\delta}$, with probability at least $n^{-\frac{1}{k}}$, the balls $B_X(u,\frac{\delta}{2k(2k+\delta)}\Delta),B_X(v,\frac{\delta}{2k(2k+\delta)}\Delta)$ contained in a single cluster.
\end{lemma}
For the case where $\delta=0$, is a distribution formerly constructed by the first author \cite{Fil19paddedArxiv}.\footnote{This is the full version, see also the conference version \cite{Fil19padded}.}
Our proof here follows the steps of \cite{Fil19paddedArxiv} (which is based on the \cite{CKR01} partition).
\begin{proof}[Proof of \Cref{lm:pairwise-partition-general}.]	
	Pick u.a.r. a  radius $r\in\{\frac1k,\frac2k,\dots,\frac kk\}$, and a random permutation $\pi=\{v_1,v_2,\dots,v_n\}$ over the points. Then set $C_i=B_X(v_i,r\cdot\frac{\Delta}{2})\setminus\cup_{j<i}B_X(v_j,r\cdot\frac{\Delta}{2})$.
	As a result we obtain a $\Delta$ bounded partition $\{C_i\}_{i=1}^n$.	
	
	Let $T=B_X(u,\frac{\delta}{2k(2k+\delta)}\cdot\Delta)\cup B_X(u,\frac{\delta}{2k(2k+\delta)}\cdot\Delta)$.
	Note that for every pair of points $x,y\in T$, by triangle inequality it holds that 
	\[
	d_{X}(x,y)\le d_{X}(u,v)+\frac{2\delta}{2k(2k+\delta)}\cdot\Delta\le\left(\frac{1}{2k+\delta}+\frac{2\delta}{2k(2k+\delta)}\right)\cdot\Delta=\frac{\Delta}{2k}~.
	\]
	Let $A_s=\{v_j\mid d_X(v_j,T)\le\frac{s}{k}\cdot\frac{\Delta}{2}\}$. Then $A_0=T$.
	Suppose that $r=\frac sk$, and let $v_i$ be the vertex with minimal index such that $d_X(T,v_i)\le\frac sk\cdot\frac\Delta2$. Then no vertex in $T$ will join the clusters $C_1,\dots,C_{i-1}$, and some vertex in  $T$ will join $C_i$. 
	Let $z\in T\cap C_i$, and suppose further that $v_i\in A_{s-1}$. 
	By the triangle inequality, for every $y\in T$, $d_{X}(y,v_{i})\le d_{X}(y,z)+d_{X}(z,v_{i})\le\frac{\Delta}{2k}+\frac{s-1}{k}\cdot\frac{\Delta}{2}=\frac{s}{k}\cdot\frac{\Delta}{2}	$. Hence all the points in $T$ will join the cluster of $v_i$. 
	Denote by $\Psi$ the event that all the vertices in $T$ are contained in a single cluster. Using the law of total probability, we conclude
	\[
	\Pr[\Psi]=\frac{1}{k}\cdot\sum_{s=1}^{k}\Pr[\Psi\mid r=\frac{s}{k}]\ge\frac{1}{k}\cdot\sum_{s=1}^{k}\frac{|A_{s-1}|}{|A_{s}|}\ge\left(\Pi_{s=1}^{k}\frac{|A_{s-1}|}{|A_{s}|}\right)^{\frac{1}{k}}=\left(\frac{|A_{0}|}{|A_{k}|}\right)^{\frac{1}{k}}\ge n^{-\frac{1}{k}}~,
	\]
	where the second inequality follows by the inequality of arithmetic and geometric means.	
\end{proof}

We now ready to prove \Cref{lem:pairwise-partition-general} (restated for convenience).
\pairwisePartitionGeneral*
\begin{proof}
	Fix $\Delta$.
	Sample $s=n^{\frac{1}{k}}\cdot2\ln n$ i.i.d. partitions using \Cref{lm:pairwise-partition-general}. Consider a pair of points $u,v$ such that $d_X(u,v)\le \frac{\Delta}{2k+\delta}$. Then in each sampled partition, the probability that the balls $B_X(u,\frac{\delta}{2k(2k+\delta)}\Delta),B_X(v,\frac{\delta}{2k(2k+\delta)}\Delta)$ contained in a single cluster is at least $p=n^{-\frac{1}{k}}$. The probability that $u,v$ are not satisfied by any partition, is at most $(1-p)^{s}\le e^{-ps}=e^{-2\ln n}=n^{-2}$. As there are at most ${n\choose2}\le\frac{n^2}{2}$ pairs at distance at most $\frac{\Delta}{2k+\delta}$, by union bound, with probability at least $\frac12$, every pair is satisfied by some partition. It follows that the union of $s$ random partitions is, with a probability at least $\frac12$, an $(O(n^{\frac{1}{k}}\log n), 2k+\delta,\frac{\delta}{2k(2k+\delta)})$-pairwise partition cover as required.
\end{proof}

\subsection{Ultrametric cover for doubling spaces}\label{subsec:UltraCoverDoubling}

In this section, we will construct a pairwise partition cover for doubling spaces, and then use them to construct ultrametric covers, and thus proving \Cref{thm:DoublingUltrametricCoverMain}.
We begin with the following combinatorial lemma.
\begin{lemma}
	\label{lem:MatchingCoverRedBlue}Consider a graph $G=(V,E_{b}\cup E_{r})$ with
	disjoint sets of blue edges $E_{b}$ and red edges $E_{r}$, such the maximal blue degree
	is $\delta_{b}\ge1$, and the maximal red degree is $\delta_{r}\ge1$. Then
	 there is a set of at most $\gamma = O(\delta_{r}\delta_{b})$ matching $\mathcal{M} = \{M_{1},M_{2},\dots,M_{\gamma}\}$
	of $G$ such that (a) $E_{b}\subseteq\cup_{i=1}^{\gamma}M_{i}$, and (b) for every matching $M\in \mathcal{M}$, there is no red edge whose both endpoints are matched by $M$.
\end{lemma}
\begin{proof}
	We construct $\mathcal{M}$ greedily. Initially, $\mathcal{M} = \emptyset$.	Let $E_{b}'$ be the set of blue edges of $G$ that are not added to any matching in $\mathcal{M}$. Let $M\subseteq E_{b}'$ be a maximal matching such that there is no red edge whose endpoints are both matched by $M$ (such maximal matching could be found greedily in linear time); we add $M$ to $\mathcal{M}$ and repeat.

	We argue by contradiction that the greedy algorithm adds at most $4\delta_{r}\delta_{b}$ matching to $\mathcal{M}$. Consider a vertex $v$ such that after $\delta_{b}(2\delta_{r}+2)$
	maximal matchings added to $\mathcal{M}$, there remains at least one blue edge incident to $v$ that is not covered by any matching in $\mathcal{M}$. Since there is at most $\delta_{b}$ blue edges incident to $v$, there must be a set $\mathcal{M}_v \subseteq \mathcal{M}$ of  at least $\delta_{b}(2\delta_{r}+1)$ matchings where $v$ is not matched by any of the matchings in $\mathcal{M}_v$. By the maximality, in each matching $M\in \mathcal{M}_v$, either: 
	\begin{enumerate}[noitemsep]
		\item[(a)] A red neighbor of $v$ is matched by $M$.
		\item[(b)] For every blue neighbor $u$ of $v$, either $u$ is matched, or a red neighbor of $u$ is matched by $M$, which prevents $u$ from being matched.
	\end{enumerate}
	Since $v$ has at most $\delta_{r}$ red neighbors, and each of them can be matched at most $\delta_{b}$ times, case (a) happens at most $\delta_{b}\delta_{r}$ times. 
	The blue neighbors of $v$ could be matched at most $\delta_{b}-1$ times, while their red neighbors could be matched at most $\delta_{r}\delta_{b}$ times. Thus, case (b) happens at most $\delta_{b}-1+\delta_{r}\delta_{b}<\delta_b(\delta_r+1)-1$ times.
	We conclude that $|\mathcal{M}_{v}|\leq\delta_{r}\delta_{b}+\delta_{b}(\delta_{r}+1)-1=\delta_{b}(2\delta_{r}+1)-1$, a contradiction.
\end{proof}

\begin{lemma}\label{lm:partition-doubling-pairwise}
	Every metric space $(X,d_X)$ with doubling dimension $d$ admits an  $(\epsilon^{-O(d)}, (1+\epsilon),\eps)$-pairwise partition cover scheme for any $\epsilon \in (0,1/16)$.
\end{lemma}
\begin{proof}
	Let $\Delta > 0$ be any given real number. We show that $(X,d_X)$ admits an $(\epsilon^{-O(d)}, (1+8\epsilon),\frac{\eps}{2}, (1+8\epsilon)\Delta)$-pairwise partition cover $\mathbb{P}$, the lemma then follows by rescaling $\eps$ and $\Delta$.
	
	Let $N$ be an $(\epsilon\Delta)$-net of $(X,d_X)$. We construct a graph $G$ with $N$ as the vertex set; there is a blue edge $(u,v)\in E_b$ in $G$ iff $d_{X}(u,v)\in\left[(1-4\eps)\frac{\Delta}{2},(1+2\eps)\Delta\right]$, 	and there is a red edge $(u,v) \in E_r$  iff $d_{X}(u,v)\le 4\eps\Delta$. 
	As $\eps<\frac{1}{12}$, the set of blue and red edges are disjoint.
	By the packing property of doubling metrics (\Cref{lem:doubling_packing}), every vertex in $G$ has blue degree $\epsilon^{-O(d)}$ and red degree $2^{O(d)}$.  Let $\mathcal{M}$ be the set of matching of $G(N,E_b\cup E_r)$ guaranteed by \Cref{lem:MatchingCoverRedBlue}; $|\mathcal{M}| = O(\epsilon^{-O(d)}2^{O(d)}) = \epsilon^{-O(d)}$. 
	
	For each matching $M\in \mathcal{M}$, we construct a partition $\mathcal{P}_{M}$ as follows: for every edge $\{u, v\}\in M$, we add $B_X(u,2\eps\Delta)\cup B_X(v,2\eps\Delta)$ as a cluster to $\mathcal{P}_M$. 
	Denote by $N_M$ the set of net points that remain unclustered. 
	For every net point $x\in N_M$, we initiate a new cluster $C_x$ containing $x$ only. Then, every remaining unclustered point $z\in X$ joins the cluster of its closest net point $x_z$ (from either $N_M$ or $N\setminus N_M$).
	See \Cref{fig:PairwisePartition} for an illustration.
	
	\begin{figure}[t]
	\centering
	\includegraphics[width=0.6\textwidth]{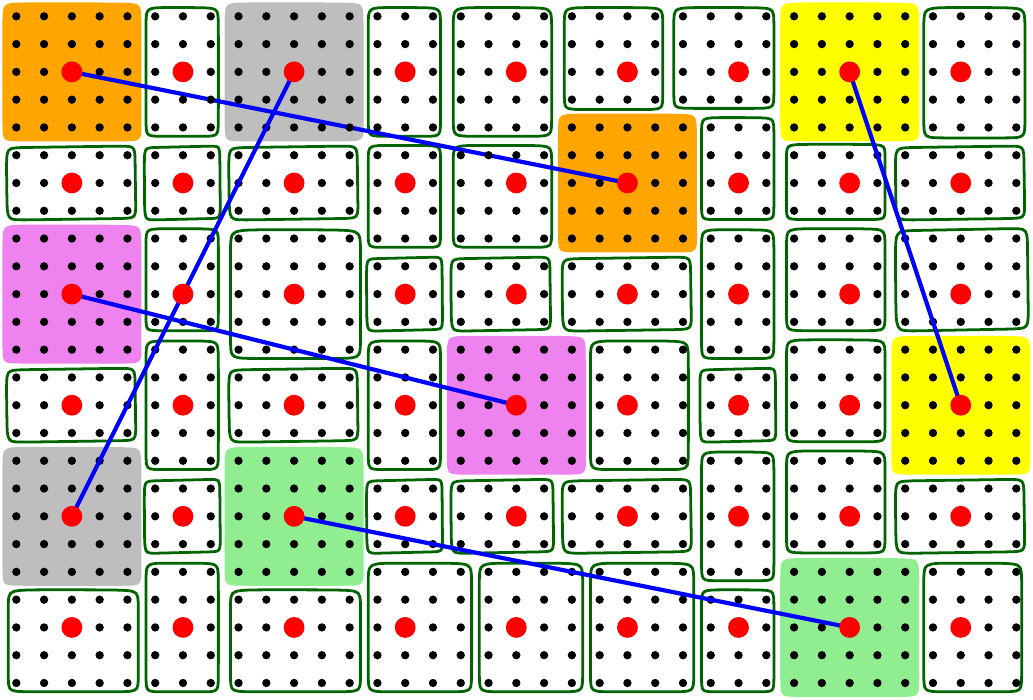}
	\caption{\footnotesize{Illustration of the partition $\mathcal{P}_M$. The black points represent metric points, while the red points represent the net points $N$. The blue edges are the matching $M$. For each edge $\{u,v\}\in M$, the cluster $B_X(u,2\eps\Delta)\cup B_X(v,2\eps\Delta)$ is added to $\mathcal{P}_M$. These clusters are illustrated with colored filled boxes.
	The remaining points are clustered around unclustered net points. These clusters are encircled by  green lines.
	Distances shown in the figure are not properly scaled for a better  visualization.
	}}
	\label{fig:PairwisePartition}
	\end{figure}
	
	We observe that for any two edges $(u,v)$ and $(u',v')$ in matching $M$, $B_X(u,2\eps\Delta)\cap B_X(u',2\eps\Delta) = \emptyset$ since otherwise, there is a red edge between $u$ and $u'$, contradicting item (b) in~\Cref{lem:MatchingCoverRedBlue}. Thus, $\mathcal{P}_M$ is indeed a partition of $X$. 
	We next bound the diameter of each cluster in $\mathcal{P}_M$. Clearly every cluster $C_x$ for $x\in N_M$ has diameter at most $2\eps\Delta$. On the other hand, by the construction and the triangle inequality, the diameter of every cluster resulting from the matching is bounded by $2\cdot (\eps\Delta+2\eps\Delta)+(1+2\eps)\Delta=(1+8\eps)\Delta$.
	Thus $\mathcal{P}_M$ is $\left((1+8\eps)\Delta\right)$-bounded.
	
	\sloppy Let $\mathbb{P} = \{\mathcal{P}_M\}_{M\in \mathcal{M}}$. 
	It remains to show that for every $x,y\in X$ such that $d_{X}(x,y)\in\left[\frac{(1+8\eps)\Delta}{(1+8\eps)2},\frac{(1+8\eps)\Delta}{(1+8\eps)}\right]=\left[\frac{\Delta}{2},\Delta\right]$, there is a cluster $C$ in a partition $\mathcal{P} \in \mathbb{P}$ containing both $B_X(x,\frac{\eps}{2}\cdot (1+8\eps)\Delta)$ and $B_X(y,\frac{\eps}{2}\cdot (1+8\eps)\Delta)$. Note that as $\eps\le\frac{1}{16}$, $\frac{\eps}{2}\cdot (1+8\eps)\Delta\le \eps\Delta$.
	Let $x',y'\in N$ be net points such that $d_X(x,x'),d_X(y,y')\le\eps\Delta$. Then by the triangle inequality
	$\left|d_{X}(x',y')-d_{X}(x,y)\right|\le2\eps\Delta$, implying that
	$d_{X}(x',y')\in\left[(1-4\eps)\frac{\Delta}{2},(1+2\eps)\Delta\right]$.
	Hence, $G$ contains a blue edge between $x',y'$. It follows that there is a matching $M$ containing the edge $\{x',y'\}$, and a partitions $\mathcal{P}_M$ containing the cluster $C=B_X(x',2\eps\Delta)\cup B_X(y',2\eps\Delta)$. In particular, $B_X(x,\eps\Delta)\cup B_X(y,\eps\Delta)\subseteq C$ as required.
\end{proof}

We are finally ready to prove \Cref{thm:DoublingUltrametricCoverMain} that we restate below for convenience.
\DoublingUltrametricCover*
\begin{proof}
	We begin with the first assertion (doubling metrics admit ultrametric covers).
	After appropriate rescaling, by \Cref{lem:PairwisePartitionCoverSchemeToUltrametric} and 
	\Cref{lm:partition-doubling-pairwise}, we obtain an
	$\left(\eps^{-O(d)},1+\epsilon\right)$-ultrametric cover where every ultrametric in the cover is a $\frac{1}{\epsilon}$-HST.
	It remains to show that every ultrametric in the cover returned by \Cref{lem:PairwisePartitionCoverSchemeToUltrametric} w.r.t. the pairwise partition cover scheme constructed in \Cref{lm:partition-doubling-pairwise} has bounded degree.
	
	We will use the terminology of \Cref{lem:PairwisePartitionCoverSchemeToUltrametric,lm:partition-doubling-pairwise}.
	Consider some ultrametric $\mathcal{U}_j$ and some cluster $\tilde{C}$ at level $i$ with label $(1+\eps) \Delta_i$. 
	The clusters at level $i-1$ correspond to points in $(\eps\Delta_{i-1})$-net. The cluster $\tilde{C}$ has diameter $(1+\eps)\Delta_i$, and hence it has at most $\eps^{-O(d)}$ $(\epsilon\Delta_{i-1})$-net points. 
	In particular $\tilde{C}$ can contain at most $\eps^{-O(d)}$ level-$(i-1)$ clusters. The bound on the degree follows.
	
	Next, we prove the second assertion. Consider a metric space $(X,d_X)$ admitting a $(\tau,\rho,k,\delta)$-ultrametric cover with $k\geq 2\rho$. Let $B_X(x,r)$ be some ball of radius $r$. In each ultrametric $U_i$ in the cover, let $L_i$ be the node closest to root that is an ancestor of $x$ and has label at most $\rho\cdot r$. Let $\{L_{i,1},L_{i,2},\dots\}$ be the set of at most $\delta$ children of $L_i$ in $U_i$. For each $L_{i,j}$, we pick an arbitrary leaf $u_{i,j}\in X$ descendent of $L_{i,j}$. We argue that 
	$$B_X(x,r)\subseteq \cup_{i,j}B_X(u_{i,j},\frac r2)~,$$
	as the number of balls in the union is at most $\tau\delta$, the theorem will follow.
	Consider a vertex $y\in B_X(x,r)$. There is necessarily an ultrametric $U_i$ such that $d_{U_i}(x,y)\le \rho\cdot r$. In particular, in $U_i$, $x,y$ are both decedents of a node with label at most $\rho\cdot r$. Recall that $L_i$ is such a node with maximal label. Let $L_{i,j}$ be the child of $L_i$ such that $y$ is decedent of $L_{i,j}$. As $U_i$ is a $k$-HST, the label of $L_{i,j}$ is bounded by $\frac{\rho r}{k}\le\frac r2$ since $k\geq 2\rho$. In particular, $y\in B_X(u_{i,j},\frac r2)$ as required.
\end{proof}

%% file: LSO.tex
\section{Locality Sensitive Ordering} \label{sec:LSODoubling}

\paragraph{Locality-Sensitive Ordering.~} Chan \etal \cite{CHJ20} introduced and studied the notion of locality-sensitive ordering (\Cref{def:LSOclassic}).
In the same paper, Chan \etal \cite{CHJ20} showed that the Euclidean metric of dimension $d$ has an $\left(O(\epsilon)^{-d}\log\frac{1}{\epsilon}),\epsilon\right)$-LSO. They also presented various applications of the LSO to solve fundamental geometry problems in Euclidean spaces.  The proof relies on the following lemma by Walecki~\cite{Alspach08}.

\begin{lemma}\label{lm:AdjOrdering} Given a set of $n$ elements $[n] = \{1,\ldots,n\}$, there exists a set $\Sigma$ of $\lceil \frac{n}{2} \rceil$ orderings such that for any two elements $i\not= j\in [n]$, there exists an ordering $\sigma$ in which $i$ and $j$ are adjacent.
\end{lemma}

We show that metrics admitting an ultrametric cover of bounded degree have an LSO with a small number of orders. Our proof relies on the following lemma.

\begin{lemma}\label{lm:UltrametricOrdering}
	Every $\alpha$-HST $\left(U,d_{U}\right)$ of degree $\delta$ admits a
	$(\left\lceil \frac{\delta}{2}\right\rceil ,\frac{1}{\alpha})$-LSO.
\end{lemma}
\begin{proof}
	For simplicity, we will assume that the number of children in each
	node is exactly $\delta$. This could be achieved by adding dummy nodes.	By \Cref{lm:AdjOrdering}, every set of $\delta$ vertices can be ordered into $\left\lceil \frac{\delta}{2}\right\rceil $
	orderings such that every two vertices are adjacent in at least one of them. Denote these orderings by $\sigma_{1},\dots,\sigma_{\left\lceil \frac{\delta}{2}\right\rceil }$. We construct the set of orderings for $(U,d_U)$ inductively.
	
	Let $A$ be the root of the HST with children
	$A_{1},\dots,A_{\delta}$. By the induction hypothesis, each $A_{i}$
	admits a set $\sigma_{1}^{i},\dots,\sigma_{\left\lceil \frac{\delta}{2}\right\rceil }^{i}$
	orderings. We construct $\lceil \delta/2 \rceil$ orderings as follows: for each $j \in [\lceil \delta/2 \rceil]$, order
	the vertices inside each $A_{i}$ w.r.t. $\sigma_{j}^{i}$, and order
	the sets in between them w.r.t. $\sigma_{j}$. The resulting ordering is denoted by 
	$\tilde{\sigma}_{j}$. This finishes the construction.
	
	Next, we argue that this is a $(\left\lceil \frac{\delta}{2}\right\rceil ,\frac{1}{\alpha})$-LSO.
	Clearly, we used exactly $\left\lceil \frac{\delta}{2}\right\rceil $
	orderings. Let $\Delta$ be the label of the root. Consider a pair
	of leaves $x,y$. If $d_{U}(x,y)<\Delta$, then there is some $i$
	such that $x,y\in A_{i}$. By the induction hypothesis, there is
	some order $\sigma_{j}^{i}$ of $A_{i}$ such that (w.l.o.g.) $x\prec_{\sigma_{j}^{i}}y$,
	and the points between $x$ and $y$ w.r.t. $\sigma_{j}^{i}$ could
	be partitioned into two consecutive intervals $I_{x},I_{y}$ such
	that $I_{x}\subseteq B_{U}(x,d_{U}(x,y)/\alpha)$ and $I_{y}\subseteq B_{U}(y, d_{U}(x,y)/\alpha)$. The base case is trivial since every leaf has label $0$.
	Note that $\sigma_{j}^{i}$ is a sub-ordering of $\tilde{\sigma}_{j}$.
	In particular, the lemma holds.
	
	The next case is when $d_{U}(x,y)=\Delta$. Then there are $i\ne i'$
	such that $x\in A_{i}$ and $y\in A_{i'}$. There is some ordering
	$\sigma_{j}$ such that $A_{i}$ and $A_{i'}$ are consecutive. In
	particular all the vertices between $x$ to $y$ in $\tilde{\sigma}_{j}$
	can be partitioned to two sets, the first belonging to $A_{i}$, and
	the second to $A_{i'}$. The lemma follows as all the vertices in $A_{i}$
	($A_{i'}$) are at distance at most $\frac{\Delta}{\alpha}=\frac{d_{U}(x,y)}{\alpha}$
	from $x$ ($y$).
\end{proof}

\begin{restatable}{lemma}{CoverToLSOClassic}\label{lem:CoverToLSOClassic} 
	If a metric $(X,d_X)$ admits a  $\left(\tau, \rho, k, \delta\right)$-ultrametric cover, then it has a $\left(\tau\cdot \lceil \frac{\delta}{2}\rceil, \frac{\rho}{k}\right)$-LSO.
\end{restatable}
\begin{proof}
Let $\mathcal{U}$ be an ultrametric cover for $(X,d_X)$. For each ultrametric $(U,d_U)$, let $\Sigma_U$ be the set of orderings obtained by applying~\Cref{lm:UltrametricOrdering}. Let $\Sigma = \cup_{U\in \mathcal{U}}\Sigma_U$. We show that $\Sigma$ is the LSO claimed by the lemma. Clearly, it contains at most $\tau\cdot\lceil\frac{\delta}{2}\rceil$ orderings.

Consider two points $x\not=y\in X$. Let $U$ be an ultrametric in $\mathcal{U}$ such that $d_U(x,y)\leq \rho\cdot d_X(x,y)$. By~\Cref{lm:UltrametricOrdering}, there is an ordering $\rho \in \Sigma_U$ such that (w.l.o.g) $x\prec_{\sigma} y$ and points between $x$ and $y$ w.r.t $\sigma$ can be partitioned into two consecutive intervals $I_x$ and $I_y$ where $I_x\subseteq B_{U}(x,\frac{d_{U}(x,y)}{k})$ and $I_y\subseteq B_{U}(y,\frac{d_{U}(x,y)}{k})$. Since  $d_U(x,y)\leq \rho\cdot d_X(x,y)$, we conclude that $I_{x}\subseteq B_{X}(x,\frac{\rho}{k}\cdot d_{X}(x,y))$ and $I_y\subseteq B_{X}(y,\frac{\rho}{k}\cdot d_{X}(x,y))$ as desired. 	
\end{proof}

Using \Cref{thm:DoublingUltrametricCoverMain} and \Cref{lem:CoverToLSOClassic} with $\rho = (1+\epsilon)$ and $k = O(\frac{1}{\epsilon})$, we conclude
\begin{restatable}[LSO For Doubling Metrics]{corollary}{DoublingLSO}
	\label{cor:DoublingLSO}
	For every $\eps$ sufficiently smaller than $1$, every metric space $(X,d_{X})$ of doubling dimension $d$ admits an $\left(\eps^{-O(d)},\epsilon\right)$-LSO.
\end{restatable}

\sloppy\paragraph{Reliable $(1+\eps)$-Spanners from LSO.~}   Buchin \etal \cite{BHO19,BHO20} constructed reliable $(1+\epsilon)$ spanners for Euclidean metrics using $(\tau(\eps),\epsilon)$-LSO by Chan \etal \cite{CHJ20}. 
Specifically, their spanner for the deterministic case has $n\cdot O(\epsilon)^{-7d}\log^{7}(\frac{1}{\epsilon})\cdot\nu^{-6}\log n(\log\log n)^{6}$ edges, while for the oblivious case, they constructed a spanner with an almost linear number of edges: $n\cdot O(\epsilon)^{-2d}\log^{3}\frac{1}{\epsilon}\cdot\nu^{-1}\log\nu^{-1}(\log\log n)^{2}\log\log\log n$.
Their key idea is to reduce the problem to the construction of reliable $(1+\epsilon)$-spanners for the (unweighted) path graph $P_n$ with $n$ vertices.  We observe that their construction of the reliable $(1+\eps)$spanners did not use any property of the metric space other than the existence of an LSO. 

\begin{theorem}[\cite{BHO19,BHO20}, implicit]\label{METAthm:BHO-SpannerUsingLSO}
	Suppose that for any $\epsilon\in(0,1)$, an $n$-point metric space $(X,d_{X})$
	admits a $(\tau(\epsilon),\epsilon)$-LSO for some function $\tau:(0,1)\rightarrow\mathbb{N}$.
	Then for every $\nu\in(0,1)$ and $\epsilon\in(0,1)$:
	\begin{enumerate}
		\item $(X,d_{X})$ admits a deterministic $\nu$-reliable $(1+\epsilon)$-spanner with\\ $n\cdot   O\left(\left(\tau(\frac{\epsilon}{c_d})\right)^{7}\nu^{-6}\log n(\log\log n)^{6}\right)$ edges for some universal constant $c_d$.
		\item  $(X,d_{X})$
		admits an oblivious $\nu$-reliable $(1+\epsilon)$-spanner
		with\\ $n\cdot O\left(\left(\tau(\frac{\epsilon}{c_0})\right){}^{2}\nu^{-1}(\log\log n)^{2}\cdot\log\left(\frac{\tau(\frac{\epsilon}{c_0})\log\log n}{\nu}\right)\right)$ edges for some universal constant $c_0$.
	\end{enumerate}
\end{theorem}
By \Cref{METAthm:BHO-SpannerUsingLSO} and \Cref{cor:DoublingLSO}, we have:
\begin{corollary}\label{cor:Doubling}
	 Consider a metric space $(X,d_X)$ with doubling dimension $d$.
	Then for every $\nu\in(0,1)$ and $\epsilon\in(0,1)$, $(X,d_{X})$
	admits a deterministic $\nu$-reliable $(1+\epsilon)$-spanner
	with $n\cdot\epsilon{}^{-O(d)}\nu^{-6}\log n(\log\log n)^{6}$ edges, 
	and an oblivious $\nu$-reliable $(1+\epsilon)$-spanner
	with\\ $n\cdot\epsilon{}^{-O(d)}\cdot\nu^{-1}(\log\log n)^{2}\cdot\log\left(\frac{\log\log n}{\epsilon\nu}\right)=n\cdot\epsilon{}^{-O(d)}\cdot\tilde{O}\left(\nu^{-1}(\log\log n)^{2}\right)$
	edges.	
\end{corollary}

%% file: triangleLSO.tex
\section{Triangle Locality Sensitive Ordering}\label{sec:triangleLSO}

A triangle locality-sensitive ordering (triangle-LSO) is defined as follows. 

\begin{definition}[$(\tau,\rho)$-Triangle-LSO]\label{def:TriangleLSO}
	Given a metric space $(X,d_{X})$, we say that a collection $\Sigma$
	of orderings is a $(\tau,\rho)$-triangle-LSO if
	$\left|\Sigma\right|\le\tau$, and for every $x,y\in X$, there is
	an ordering $\sigma\in\Sigma$ such that (w.l.o.g.)  $x\prec_{\sigma}y$, and for every $a,b\in X$ such that $x\preceq_{\sigma}a\preceq_{\sigma}b\preceq_{\sigma}y$ it holds that $d_X(a,b)\le\rho\cdot d_X(x,y)$.
\end{definition}
Note that every $(\tau,\rho)$-triangle LSO is also a $(\tau,\rho)$-LSO; however, a $(\tau,\rho)$-LSO is a  $(\tau,2\rho+1)$-triangle LSO (by triangle inequality).
Hence for stretch parameter $\rho>1$, triangle-LSO is preferable to the classic LSO.
Similar to~\Cref{lem:CoverToLSOClassic}, we show that a metric space admitting an ultrametric cover has a triangle-LSO with a small number of orderings.

\begin{restatable}{lemma}{CoverToTriangleLSO}
	\label{thm:CoverToTriangleLSO} If a metric $(X,d_X)$ admits a $(\tau,\rho)$-ultrametric cover $\mathcal{U}$, then it has a $\left(\tau, \rho\right)$-triangle-LSO.
\end{restatable}
\begin{proof}
	Let $(U,d_U)$ be an ultrametric in the cover $\mathcal{U}$. Note that $U$ is also a $1$-HST. We  create a single ordering for $\sigma_U$  by simply putting the leaves (only) in the pre-order fashion. That is, given a node $v$ with children $v_1,\dots,v_\delta$, we recursively put all the descendants of $v_1$, then $v_2$, and so on until we put all the descendants of $v$ to $\sigma_{U}$. It holds that for every $x\prec_{\sigma_U} y$ with lca $z$, all the vertices $a,b$ such that
	$x\prec_{\sigma_{U}} a\prec_{\sigma_{U}} b\prec_{\sigma_{U}} y$ are also descendants of $z$. Thus, $d_U(a,b)\leq d_U(x,y)$.
	
	We now show that $\{\sigma_U\}_{U\in \mathcal{U}}$ is a $(\tau,\rho)$-triangle-LSO.  For any given $x,y$, there exists $U\in \mathcal{U}$ such that $d_X(x,y)\leq d_U(x,y)\leq \rho\cdot d_X(x,y)$. Then for every pair of vertices $a,b$ such that
	$x\prec_{\sigma_{U}} a\prec_{\sigma_{U}} b\prec_{\sigma_{U}} y$, $d_X(a,b)\leq d_{U}(a,b)\leq d_U(x,y)\leq \rho\cdot d_X(x,y)$ as desired.
\end{proof}

Using \Cref{thm:GeneralUltrametricCover} and \Cref{thm:CoverToTriangleLSO}, we conclude.
\begin{corollary}\label{cor:TriangleLSOGeneralGraphs}
		\sloppy For every $k\in\N$, and $\eps\in(0,1)$, every $n$-point metric space admits an $\left(O(n^{\frac{1}{k}}\cdot\log n\cdot\frac{k^{2}}{\eps}\cdot\log\frac{k}{\eps}),2k+\eps\right)$-triangle-LSO.
\end{corollary}

\paragraph{Reliable Spanners from triangle-LSO.~} We show that if a metric space admits a $(\tau,\rho)$-triangle-LSO, it has an oblivious $2\rho$-spanners with about $n\cdot\tau^2$ edges. We use this result to construct reliable  $(8k-2)(1+\epsilon)$-spanners for  general metrics (and reliable $(2\rho)$-spanners for ultrametrics).

\begin{restatable}{theorem}{TriangleLSOtoSpanner}\label{Thm:TriangleLSOtoSpanner} Suppose that a metric space $(X,d_{X})$ admits a $(\tau,\rho)$-Triangle-LSO. Then for every $\nu\in(0,1)$, $X$ admits an \emph{oblivious} $\nu$-reliable, $(2\rho)$-spanner with $n\tau\cdot O\left(\log^{2}n+\nu^{-1}\tau\log n\cdot\log\log n\right)$ edges.
\end{restatable}

The proof of~\Cref{Thm:TriangleLSOtoSpanner} is deferred to~\Cref{subsec:TriangleLSOToSpanner}.  Using \Cref{cor:TriangleLSOGeneralGraphs} with parameters $2k$ and $\frac\eps2$, and \Cref{Thm:TriangleLSOtoSpanner} we conclude:
\begin{theorem}[Oblivious Reliable Spanner for General Metric]\label{thm:generalMetricOblivious}
	For every $n$-point metric space $(X,d)$ and parameters $\nu\in(0,1)$, $\eps\in(0,\frac12)$, $k\in\N$, $(X,d)$ admits an oblivious $\nu$-reliable, $8k+\eps$-spanner for $X$ with $n^{1+\frac{1}{k}}\cdot\nu^{-1}\cdot\log^{3}n\cdot\log\log n\cdot\frac{k^{4}}{\eps^{2}}O(\log\frac{k}{\eps})^{2}=n^{1+\frac{1}{k}}\cdot\nu^{-1}\cdot\tilde{O}\left(\log^{3}n\cdot\frac{k^{4}}{\eps^{2}}\right)$
	 edges.
\end{theorem}

By~\Cref{thm:CoverToTriangleLSO} and~\Cref{Thm:TriangleLSOtoSpanner}, we obtain:
\begin{corollary}\label{cor:UltrametricSpanner}
	For every parameter $\nu\in(0,1)$, every $n$-point ultrametric space $(X,d)$ admits an oblivious $\nu$-reliable, $2$-spanner  with $n\cdot O\left(\log^{2}n+\nu^{-1}\log n\cdot\log\log n\right)$ edges.
\end{corollary}
The stretch parameter in \Cref{cor:UltrametricSpanner} is tight; see \Cref{rem:ultrametricLB}.

For the rest of this section, we show how to construct reliable spanners for metric spaces admitting a triangle LSO.  Following the approach of Buchin \etal\cite{BHO19,BHO20}, we reduce the problem to the construction of reliable spanners for the (unweighted) path graph $P_n$.  However, in our setting, we face a very different challenge: when $\rho > 1$, the \emph{stretch} of the reliable spanner for $(X,d_X)$ grows linearly w.r.t. the \emph{number of hops} of the reliable spanner for the path graph. In the Euclidean setting studied by Buchin \etal \cite{BHO20}, the stretch  parameter is $\rho = \epsilon < 1$, and as a result, the stretch is not significantly affected by the number of hops of the spanner for the path graph. 

Buchin \etal \cite{BHO20} constructed an oblivious $\nu$-reliable $1$-spanner for the path graph $P_n$  with $O(n\cdot\nu^{-1}\log \nu^{-1})$ edges. However, their spanner has hop diameter $2\log n$; if we use their reliable spanner for the path graph, we will end up in a spanner for $(X,d_X)$ with stretch $2\rho\log n$.
To have a stretch $2\rho$, we construct a reliable spanner for the path graph $P_n$ with hop diameter $2$.
As a consequence, the sparsity of our spanner has some additional logarithmic factors. 
Note that a $2$-hop spanner for the path graph $P_n$ even without any reliability guarantee must contain $\Omega(n\log n)$ edges (see Excercise 12.10~\cite{NS07}).
Our result is summarized in the following lemma whose proof is deferred to~\Cref{subsec:2HopSpanner}.
\begin{restatable}[$2$-hop-spanner]{lemma}{TwoHopSpanner}
	\label{lem:2-hopSpanner}
	For every $\nu\in(0,1)$, the path graph $P_{n}$ admits an
	oblivious $\nu$-reliable, $2$-hop $1$-spanner $H$ with $n\cdot O\left(\log^{2}n+\nu^{-1}\log n\cdot\log\log n\right)$
	edges.	
\end{restatable}

Using \Cref{lem:2-hopSpanner}, we can construct a reliable spanner for metric spaces admitting a $(\tau,\rho)$-triangle LSO as claimed by~\Cref{Thm:TriangleLSOtoSpanner}.

\subsection{Proof of~\Cref{Thm:TriangleLSOtoSpanner}}\label{subsec:TriangleLSOToSpanner}

	Let $\Sigma$ be a $(\tau,\rho)$-triangle LSO as assumed by the theorem. Let $\nu'=\frac{\nu}{\tau}$. 
	For every ordering $\sigma\in \Sigma$, we form an unweighted path graph $P_{\sigma}$ with vertex set $X$ and the order of vertices along the path is $\sigma$. We construct a $\nu'$-reliable $2$-hop spanner $H_\sigma(X,E_{\sigma},w_{\sigma}))$ for $P^{\sigma}_{n}$ with $n\cdot O\left(\log^{2}n+\nu'^{-1}\log n\cdot\log\log n\right)$ edges by~\Cref{lem:2-hopSpanner}. Note that for every edge $\{u,v\}\in E_{\sigma}$ of $H_{\sigma}$,  $w_{\sigma}(u,v)$ is the distance between $u$ and $v$ in the (unweighted) path graph $P^{\sigma}_{n}$. 
	
	We form a new weight function $w_{X}$ that assign each edge $\{u,v\}\in E_{\sigma}$ a weight $w_X(u,v) = d_X(u,v)$. The reliable spanner for $(X,d_X)$ is $H = \bigcup_{\sigma\in \Sigma}H_{\sigma}(X,E_{\sigma},w_X)$. We observe that the total number of edges in $H$ is bounded by
	\[
	\sum_{\sigma\in\Sigma}n\cdot O\left(\log^{2}n+\nu'^{-1}\log n\cdot\log\log n\right)=n\tau\cdot O\left(\log^{2}n+\nu^{-1}\tau\log n\cdot\log\log n\right)~.
	\]
	
	Let $B\subseteq X$ be an oblivious attack. Let $B^+_{\sigma}$ be the faulty extension of $B$ in $H_{\sigma}$, and $B^+ = \cup_{\sigma\in \Sigma}B^+_{\sigma}$ be the faulty extension of $B$ in $H$. We observe that:
	\[
	\mathbb{E}\left[\left|B^{+}\right|\right]\le|B|+\sum_{\sigma}\mathbb{E}\left[\left|B_{\sigma}^{+}\setminus B\right|\right]\le|B|+\tau\nu'\cdot|B|\le(1+\nu)\cdot|B|~.
	\]
	
	It remains to show the stretch guarantee of $H$.  For every pair of points $x,y\notin B^+$, let $\sigma \in \Sigma$ be the ordering that satisfies the ordering property for $x$ and $y$: for every $a,b\in X$ such that $x\preceq_{\sigma} a \preceq_{\sigma} b \preceq_{\sigma} y$, $d_X(a,b)\leq \rho\cdot d_X(x,y)$. (Here we assume~w.l.o.g. that $x\preceq_{\sigma}y$.) Since $x,y\notin B^+$, $x,y \notin B_{\sigma}^+$. Since $H_\sigma(X,E_{\sigma},w_{\sigma}))$ is a $2$-hop $1$-spanner for $P_{\sigma}$, there must be $z\not\in B$ such that $x\preceq_{\sigma}z\preceq_{\sigma}y$ and $\{x,z\},\{z,y\}\in E_\sigma$. We conclude that 
	\[
	d_{H}(x,y)\le w_X(x,z) + w_X(z,y) = d_{X}(x,z)+d_{X}(z,y)\le2\rho\cdot d_{X}(x,y)~;
	\]
	the theorem follows.

\subsection{$2$-hop spanner: Proof of \Cref{lem:2-hopSpanner}}\label{subsec:2HopSpanner}

In this section, we construct a 2-hop reliable spanner for the path graph $P_n$.  Our construction uses the concept of a \emph{shadow} introduced by Buchin, \etal \cite{BHO19}.

\begin{definition}\label{def:shaddow}
	Let $B$ be a subset of $[n]$.
	The \emph{left $\alpha$-shadow} of $B$ is all the vertices $b$ such for some $a<b$, $|[a:b]\cap B|\ge\alpha\cdot|[a:b]|$, denoted by $\mathcal{S}_L(\alpha,B)$. The \emph{right $\alpha$-shadow} $\mathcal{S}_R(\alpha,B)$ is defined symmetrically. The set  $\mathcal{S}=\mathcal{S}_L(\alpha,B)\cup \mathcal{S}_R(\alpha,B)$ is called the \emph{$\alpha$-shadow} of $B$.
\end{definition}

Buchin \etal \cite{BHO19} showed that for every $\alpha\in (\frac23,1)$ and any $B\subseteq [n]$, $|\mathcal{S}_L(\alpha,B)|\le\frac{|B|}{2\alpha-1}$. In particular, for $\beta<\frac13$, 
\begin{equation}
|\mathcal{S}(1-\beta,B)\setminus B|\le\frac{|B|}{2(1-\beta)-1}-|B|\le(1+6\beta)|B|-|B|=6\beta\cdot|B|~.\label{eq:shaddow}
\end{equation}

We also use a construction of a (non-reliable) $2$-hop $1$-spanner for path graph $P_n$ with $O(n\log n)$ edges~\cite{AS87,BTS94}.

\begin{lemma}[\cite{AS87,BTS94}]\label{lm:path-spanner} There is a $2$-hop $1$-spanner for $P_n$ with $O(n\log n)$ edges.
\end{lemma}

We a now ready to prove \Cref{lem:2-hopSpanner}, which we restate below for convenience.

\TwoHopSpanner*
\begin{proof}
	We will assume that $\nu>\frac1n$, as otherwise, we can take all the possible edges to the spanner. We will construct a spanner with reliability parameter $O(\nu)$; the lemma will follow by rescaling.
	
	Let $M=\log n$, and for every $j\in[0,M]$, let $p_{j}=c\left(\ln n+\nu^{-1}\ln(\log n)\right)\cdot2^{-j}$ for a constant $c$ to be determined later. We construct a vertex set $N_j$ by sampling each vertex with probability $\min\{1,p_{j}\}$. Vertices in $N_j$ are called \emph{centers}.

	The spanner $H$ is constructed as follows: For every vertex $v$, and parameter $j\in[0,M]$, we add edges from $v$ to $[v-2^{j+1}:v+2^{j+1}] \cap N_{i}$. That is to every vertex $u\in N_{j}$ such that $|v-u|\le2^{j+1}$. We also add to $H$ a $2$-hop spanner of size $O(n\log n)$ guaranteed by~\Cref{lm:path-spanner}.
	
	To bound the number of edges, we charge only to the centers; each edge between two different centers is charged twice. The expected number of edges charged to a vertex $a$ is 
	\[
	\sum_{j=0}^{M}\Pr\left[a\in N_{j}\right]\cdot2^{j+1}\cdot2=c\left(\log n+\nu^{-1}\ln\log n\right)\cdot\sum_{j=1}^{M}2^{-j}\cdot2^{j+2}=O(\log n)\cdot\left(\log n+\nu^{-1}\ln\log n\right)~,
	\]
	implying that the expected number of edges is $n\cdot O\left(\log^{2}n+\nu^{-1}\log n\cdot\log\log n\right)$. The worst-case guarantee promised in the lemma follows by using \Cref{lem:ExpectedSize}.
	
	Let $B$ be an oblivious attack. The faulty extension $B^{+}$ will consist of $B$ and all the vertices $v$ such that for some $j$, (a)  $v>2^{j}$ and  $[v-2^{j}, v]\cap N_{j}\subseteq B$ or (b) $v\le n-2^j$ and $[v, v+2^{j}]\cap N_{j}\subseteq B$. This way, for every  pair of vertices $u,v\not\in B^+$ whose distance  is in $[2^{i},2^{i+1})$, there is a vertex  $z\in N_{i}\setminus B$ such that both will have edges to $z$. This implies the monotone $2$-hop property for $u$ and $v$.

	Finally we analyze the expected size of $B^{+}$. If $B=\emptyset$, then $B^+=\emptyset$ as we added a $2$-hop spanner. 
	Otherwise, a vertex $v\notin B$ joins $B^+$ iff $N_j\cap[v:v+2^{j}]\subseteq B$ or $N_j\cap[v-2^{j}:v]\subseteq B$ for some $j\in [0,M]$.

	Fix an $i \leq \lfloor\log\frac{1}{3\nu}\rfloor$. For every vertex $a\notin \mathcal{S}_L(1-2^{i}\cdot \nu,B)$ and every $j\in [0,M]$, it holds that $[a,a+2^{j}]\cap B\le (1-2^{i}\cdot \nu)\cdot 2^j$ by the definition of a left shadow. In particular, there are at least $2^{i+j}\cdot \nu$ vertices in $[a,a+2^{j}]\setminus B$. Vertex $a$ joins $B^+$ iff for some $j$, $([a:a+2^{j}]\setminus B)\cap N_j = \emptyset$ or  $([a-2^{j}:a]\setminus B)\cap N_j = \emptyset$.  (Note that if the distance from $a$ to either $1$ or $n$ is smaller than $2^j$, then it cannot join $B^+$ by definition.)  We observe that:
	\begin{align*}
	\Pr\left[[a,a+2^{j}]\cap N_{j}\subseteq B\mid a\notin\mathcal{S}(1-2^{i}\cdot\nu,B)\right] & \le\left(1-2^{-j}\cdot c\cdot\nu^{-1}\cdot\ln(\log n)\right)^{2^{i+j}\cdot\nu}\\
	& \le e^{-2^{-j}\cdot c\cdot\nu^{-1}\cdot\ln\log n\cdot2^{i+j}\cdot\nu}=(\log n)^{-c\cdot 2^{i}}~.
	\end{align*}
	Thus, the probability that $a$ is added to $B^+$ is at most:
	\begin{align*}
	\Pr\left[a\in B^{+}\mid a\notin\mathcal{S}(1-2^{i}\cdot\nu,B)\right] & \le\sum_{j=1}^{M}2\cdot\Pr\left[[a,a+2^{j}]\cap N_{j}\subseteq B\mid a\notin\mathcal{S}(1-2^{i}\cdot\nu,B)\right]\\
	& \le\sum_{j=1}^{M}2\cdot(\log n)^{-c\cdot 2^{i}}\le2\cdot(\log n)^{-c\cdot 2^{i}-1}~.
	\end{align*}
	On the other hand, for a vertex $a\notin\mathcal{S}(\frac{2}{3},B)$, we have that 
	\[
	\Pr\left[[a,a+2^{j}]\cap N_{j}\subseteq B\mid a\notin\mathcal{S}(\frac{2}{3},B)\right]\le\left(1-2^{-j}\cdot c\cdot\ln n\right)^{\frac{1}{3}\cdot2^{j}}\le e^{-2^{-j}\cdot c\cdot\ln n\cdot\frac{1}{3}\cdot2^{j}}=n^{-\frac{c}{3}}~,
	\]
	implying that
	\begin{align*}
	\Pr\left[a\in B^{+}\mid a\notin\mathcal{S}(\frac{2}{3},B)\right] & \le\sum_{j=1}^{M}2\cdot\Pr\left[[a,a+2^{j}]\cap N_{j}\subseteq B\mid a\notin\mathcal{S}(\frac{2}{3},B)\right]\\
	& \le\sum_{j=1}^{M}2\cdot n^{-\frac{c}{3}}\le2\log n\cdot n^{-\frac{c}{3}}~.
	\end{align*}
	We conclude that
	\begin{align*}
	\mathbb{E}\left[B^{+}\right] & \le|\mathcal{S}(1-\nu,B)|+\sum_{i=1}^{\log\frac{1}{3\nu}}|\mathcal{S}(1-2^{i}\cdot\nu,B)\setminus B|\cdot\Pr\left[a\in B^{+}\mid a\notin\mathcal{S}(1-2^{i-1}\cdot\nu,B)\right]\\
	& \phantom{~\le|\mathcal{S}(1-\nu,B)|}+\left|[n]\setminus\mathcal{S}(\frac{2}{3},B)\right|\cdot\Pr\left[a\in B^{+}\mid a\notin\mathcal{S}(\frac{2}{3},B)\right]\\
	& \overset{(\ref{eq:shaddow})}{\le}(1+4\nu)|B|+\sum_{i=1}^{\log\frac{1}{3\nu}}2^{i+2}\nu\cdot|B|\cdot2\cdot(\log n)^{-c2^{i}-1}+n\cdot2\log n\cdot n^{-\frac{c}{3}}\quad=\quad(1+O(\nu))|B|~,
	\end{align*}
	for a sufficiently large constant $c$.
\end{proof}

%% file: LeftSidedLSO.tex
\section{Left-sided Locality Sensitive Ordering}\label{sec:leftsidedLSO}

We use a left-sided LSO to construct reliable spanners with \emph{optimal stretch} for trees, planar graphs, bounded treewidth graphs, and minor free graphs. We say an ordering  of $(X,d_X)$ is \emph{partial} if it is a linear ordering on a \emph{subset} of points in $X$. 
\begin{definition}[$(\tau,\rho)$-left-sided LSO]\label{def:leftsidedLSO}
	Given a metric space $(X,d_{X})$, we say that a collection $\Sigma$
	of \emph{partial orderings} is a $(\tau,\rho)$-left-sided LSO if every point $x\in X$ belongs to at most $\tau$ orderings in the collection, and for every $x,y\in X$, there is an order $\sigma\in\Sigma$ such that for every $x'\preceq_\sigma x$ and $y'\preceq_\sigma y$ it holds that $d_X(x',y')\le\rho \cdot d_X(x,y)$.
\end{definition}

Unlike the classic LSO and  the triangle-LSO, a  $(\tau,\rho)$-left-sided LSO may contain $\Omega(n)$  (partial) orderings; it only guarantees that each point belongs to at most $\tau$ (partial) orderings. 
Note that the stretch guarantee of a $(\tau,\rho)$-left-sided LSO also fulfills the stretch requirement of a $(\tau,\rho)$-LSO.
To see this, consider a pair of points $x,y\in X$, and let $\sigma\in \Sigma$ be the ordering guaranteed above, where w.l.o.g $x\preceq_\sigma y$. Then for every $z$ such that $x \preceq_\sigma z \preceq_\sigma y$, it holds that $d_X(x,z)\le \rho \cdot d_X(x,y)$ (as $x\preceq_\sigma z$ and $z\preceq_\sigma y$). In particular $z\in B_X(x,\rho\cdot d_X(x,y))$.
However, there is no $f(.)$ such that a $(\tau,\rho)$-left-sided LSO will be guaranteed to be an $(f(\tau),\rho)$-LSO as the number of orderings in a $(\tau,\rho)$-left-sided LSO, by definition, could be $\Omega(n)$ even if $\tau$ is a constant.
We show that minor-free metrics admit a left-sided LSO where each point belongs to a small number of orderings.

\begin{restatable}{theorem}{LeftLSOTreewidth}\label{thm:LeftLSOTreewidth} Treewidth-$k$ metrics admit a $(k\log n, 1)$-left-sided LSO.
In particular, tree metrics admit a $(\log n, 1)$-left-sided LSO.
\end{restatable}
\begin{restatable}{theorem}{LeftLSOMinorFree}\label{thm:LeftLSOMinorFree} For every $\eps\in(0,\frac12)$, every graph with \SPDdepth $k$ admits an $(O(\frac{k}{\epsilon}\log n),1+\eps)$-left-sided-LSO.
In particular, every graph
excluding a fixed minor (e.g. a planar graph) admits an  $(O(\frac{1}{\epsilon}\log^2 n),1+\eps)$-left-sided-LSO.
\end{restatable}
The proof of \Cref{thm:LeftLSOTreewidth} is deferred to \Cref{subsec:LeftLSOTw}.
The definition of \SPDdepth and the proof of~\Cref{thm:LeftLSOMinorFree} are deferred to~\Cref{subsec:SPD}.

\paragraph{Reliable Spanners from left-sided LSO's.~} Similar to~\Cref{Thm:TriangleLSOtoSpanner}, we can construct a reliable spanner with a nearly linear number of edges from a left-sided LSO. 

\begin{restatable}[Reliable Spanner From Left-sided LSO]{theorem}{SpannerLeftLSO}
	\label{thm:MetaLeftLSOMain}
	Consider an $n$-point metric space $(X,d_X)$ that admits a  $(\tau,\rho)$-left-sided LSO. Then for every $\nu\in(0,1)$, $X$ admits an oblivious $\nu$-reliable, $(2\rho)$-spanner with $n\cdot O(\nu^{-1}\tau^{2}\log n)$ edges.
\end{restatable}

The proof of~\Cref{thm:MetaLeftLSOMain} is deferred to~\Cref{subsec:LeftLSOToSpanners}.  By applying \Cref{thm:MetaLeftLSOMain} on \Cref{thm:LeftLSOTreewidth} and \Cref{thm:LeftLSOMinorFree} we have:

\begin{theorem}\label{thm:SpannerForTw} Treewidth-$k$ graphs admit   oblivious $\nu$-reliable $2$-spanners with $n\cdot O(\nu^{-1}k^{2}\log^{3}n)$ edges. In particular, trees admit  oblivious $\nu$-reliable $2$-spanners with $n\cdot O(\nu^{-1}\log^{3}n)$ edges. 
\end{theorem}
\begin{theorem}\label{thm:SpannerForMinorFree} 
	For every $\eps\in(0,\frac12)$, every graph with \SPDdepth $k$ admits a
	$\nu$-reliable $2(1+\eps)$-spanner with $n\cdot O(\nu^{-1}\frac{k^2}{\eps^{2}}\log^{3}n)$ edges.
	In particular, every graph
	excluding a fixed minor (e.g. a planar graph) admits a $\nu$-reliable $2(1+\eps)$-spanner with $n\cdot O(\nu^{-1}\eps^{-2}\log^{5}n)$ edges.
\end{theorem}

Finally, we observe that the stretch of reliable spanners in~\Cref{thm:SpannerForTw} and \Cref{thm:SpannerForMinorFree} is essentially optimal even for tree metrics.
\begin{observation}\label{thm:LBtree}
	For every $\nu\in(0,1)$, and $n\in \N$, there is an $n$-vertex unweighted tree $T=(V,E)$ such that every oblivious $\nu$-reliable spanner with stretch $t<2$ has $\Omega(n^{2})$ edges.
\end{observation}
\begin{proof}
	Let $T$ be the star graph. That is, there is a single vertex $r$ of degree $n-1$ vertices of degree $1$, and let $H$ be a $\nu$-reliable spanner of stretch $t$. Let $B=\{r\}$ be an attack. Every pair of vertices $u,v\notin B^+$ must be adjacent, as every path of length at least $2$ has weight $4>t\cdot d_T(u,v)$. It follows that $H$ must contain $\mathbb{E}\left[{\left|V\setminus B^{+}\right| \choose 2}\right]=\Omega\left(n-\mathbb{E}\left[|B^{+}|\right]\right)^{2}=\Omega(n^{2})$
	edges.
\end{proof}
\begin{remark}\label{rem:ultrametricLB}
	Note that \Cref{thm:LBtree} will also hold for ultrametrics, as the uniform metric on $n$ points is an ultrametric (here the attack $B$ could be $\emptyset$). In particular, the stretch parameter in \Cref{cor:UltrametricSpanner} is tight.
\end{remark}

\subsection{$k$-separable graphs}\label{subsec:LeftLSOTw}

\begin{definition}[$k$-separable graph family]
	A graph family $\mathcal{G}$ is $k$-separable if for every graph $G=(V,E,w)\in\mathcal{G}$, and every induced subgraph $G'=G[V']$ for $V'\subseteq V$,
	there is a set $K\subseteq V'$ of at most $k$ vertices such that each connected component in $G'\setminus K$ contains at most $|V'|/2$ vertices.
\end{definition}

Note that trees belong to a $1$-separable graph family and treewidth $k$ graphs belong to  a $k$-separable graph family. Hence \Cref{thm:LeftLSOTreewidth} is an immediate corollary of the following lemma.

\begin{restatable}[$k$-Separators to LSO]{lemma}{SeperatorToLSO}
	\label{lem:SeperatorToLSO}
	Let $\mathcal{G}$ be a $k$-separable family, then every $n$-vertex graph $G=(V,E,w)\in\mathcal{G}$ admits a $(k\log n,1)$-left-sided LSO.
\end{restatable}
\begin{proof} 
	
	We recursively construct a set of orderings $\Sigma$; initially, $\Sigma = \emptyset$. Let $K$ be a separator containing at most $k$ vertices such that  every connected component in $G\setminus K$ has size at most $\frac{n}{2}$.  For every $r\in K$, let $l_r$ be an order of $V$ w.r.t. distances from $r$. That is $l_r={v_1,v_2,\dots,v_n}$ where $d_G(v_1,r)\le d_G(v_2,r)\le\dots\le d_G(v_n,r)$. The set of $k$ orderings $\{l_r\}_{r\in K}$ are added to $\Sigma$.
	
	Let $\{G_1,\ldots, G_s\}$ be the set of connected components of $G\setminus K$. For each $G_i$, let $\Sigma_i$ be the set of orderings obtained by recursively applying the construction to $G_i$; each ordering in $\Sigma_i$ is a partial ordering of $G$. We then update $\Sigma \leftarrow \Sigma\bigcup (\cup_{i\in [1,s]}\Sigma_i)$. This completes the construction of $\Sigma$.

	Next, we argue that $\Sigma$ is a $(k\log n,1)$-left-sided LSO. By construction, each vertex belongs to at most $k\log n$ orderings since the depth of the recursion is at most $\log n$ and at each level of the recursion, each vertex belongs to at most $k$ partial orderings. 
	
	Finally, consider a pair of vertices $u,v\in V$. Let $P_{u,v}$ be a shortest path from $u$ to $v$ in $G$.
	Let $G'$ be the subgraph of $G$ at highest depth in the recursive construction that contains all the vertices of $P_{u,v}$.
	Let $K$ be the separator of $G'$ of size at most $k$; by definition, there is a vertex $r\in K\cap P_{u,v}$. As $G'$ contains all the vertices of $P_{u,v}$ it holds that $d_{G'}(u,v)=d_{G'}(u,r)+d_{G'}(r,v)=d_{G}(u,r)+d_{G}(r,v)$.
	Let  $l_r$ be the ordering that is constructed for $r$.  Let $u',v'$ be any vertices such that $u'\preceq_{l_r} u$ and $v'\preceq_{l_r} v$. Then: 
	\begin{align*}
	d_{G}(u',v') & \le d_{G}(u',r)+d_{G}(r,v')\le d_{G'}(u',r)+d_{G'}(r,v')\le d_{G'}(u,r)+d_{G'}(r,v)=d_{G}(u,v)~;
	\end{align*}
	the lemma follows.
\end{proof}

\subsection{$\SPD$ Graphs}\label{subsec:SPD}
This section is devoted to proving \Cref{thm:SpannerForMinorFree}.
Abraham \etal \cite{AFGN18} introduced the family of graphs with a shortest path decomposition (abbreviated \SPD) of bounded depth. 

\begin{definition}[\SPD depth]\label{def:SPD-K} A graph has an \SPD of depth $1$ if and only if it is a (weighted) path. A graph $G$ has an \SPD of depth $k \geq 2$ if there exists a \emph{shortest path} $P$, such that deleting $P$ from the graph $G$	results in a graph whose connected components all have \SPD of depth at most $k-1$. 
\end{definition}

It is known that (see~\cite{AFGN18}) that minor-free graphs have \SPD of depth $k = O(\log n)$. However, the family of bounded \SPDdepth graphs is much richer and contains dense graphs with $K_r$ as a minor, for arbitrarily large $r$.
To prove \Cref{thm:SpannerForMinorFree}, we will show that graphs with an \SPDdepth $k$ admit a $(\frac{k}{\epsilon}\log \frac{n}{\eps},1+\eps)$-left-sided LSO, the second assertion in \Cref{thm:SpannerForMinorFree} is an immediate corollary.
Our construction will rely on the following lemma, which implicitly appeared in \cite{Kle02,Thorup04} (see  \cite{BFN19Ramsey} for an explicit proof).
\begin{lemma}\label{lem:landmarks}
	Consider a weighted graph $G=(V,E,w)$ with parameter $\eps\in(0,1)$, and let $P$ be a shortest path in $G$. Then one can find for each $v\in V$ a set of vertices, called  \emph{landmarks},  $L_v$ on $P$ of size $|L_v|=O(\frac1\eps)$, such that for any $v,u\in V$ whose shortest path between them intersects $P$, there exists $x\in L_v$ and $y\in L_u$ satisfying
	$d_G(v,x)+d_P(x,y)+d_G(y,u)\le(1+\epsilon)\cdot d_G(v,u)$.
\end{lemma}

\begin{proof}[Proof of \Cref{thm:SpannerForMinorFree}]
	We will assume that $\eps>\frac{1}{n^2}$, as otherwise we can simply return ${n\choose2}$ orderings, where for every pair $u,v$, there is an ordering where $u,v$ are the first two vertices.
	We recursively construct a set of orderings $\Sigma$; initially, $\Sigma = \emptyset$. Let $P$ be a shortest path of $G$ such that every component of $G\setminus P$ has an \SPD of depth $k-1$.
	
	We construct a set of orderings for $G$ from $P$ to add to $\Sigma$ as follows. Let $\{L_v\}_{v\in V}$ be the set of landmarks provided by~\Cref{lem:landmarks} w.r.t. $P$. 	For every vertex $v\in V$, denote $l(v)=|L_v|$. 	We construct an auxiliary tree $T$ that initially contains the shortest path $P$ only. For every vertex $v\in V$, let $L_v=\{x_1,\dots,x_{l(v)}\}\subseteq P$. We will add the vertices $\{v_i\}_{i=1}^{l(v)}$, which are $l(v)$ copies of $v$, to $T$, and for every $i$ connect $v_i$ to $x_i$ with an edge of weight $d_G(v,x_i)$.	Note that $T$ has  $\sum_{v\in V}l(v)=O(\frac{n}{\eps})$ vertices. The distances in $T$ dominate the distances in the original graph, that is for every $v_i,u_j$ copies of $v,u$ respectively, it holds that $d_{T}(v_i,u_j)\ge d_G(v,u)$. Additionally, by \Cref{lem:landmarks}, for every pair of vertices $u,v$, such that  the shortest path between them intersects $P$, it holds that $\min_{i,j}d_{T}(v_i,u_j)\le(1+\epsilon)\cdot d_G(v,u)$. Let $\Sigma_T$ be an $(O(\log \frac{n}{\epsilon}),1)$-left-sided LSO for $T$ provided by~\Cref{lem:SeperatorToLSO}. 
	Let $\Sigma'_T$ be the collection of orderings $\Sigma_T$ where we delete duplicated occurrences. That is, for every ordering $\sigma\in\Sigma$ and vertex $v$, $\sigma$ might contain multiple copies of $v$: $v_{i_1},v_{i_2},\dots$,
	we will keep only the leftmost copy of $v$ in $\sigma$; the resulting ordering is denoted by $\sigma'$. As $\eps>n^{-2}$, each copy $v_i$ of a vertex $v$ appears in at most $O(\log \frac{n}{\epsilon})=O(\log n)$ orderings in $\Sigma_T$. As each vertex has $O(\frac1\eps)$ copies, we conclude that each vertex $v$ appears in at most $O(\frac1\eps\cdot\log n)$ orderings in $\Sigma'_T$.	
	
	Let $\{G_1,\ldots, G_s\}$ be the set of connected components of $G\setminus P$. For each $G_i$, let $\Sigma_i$ be the set of orderings obtained by recursively applying the construction to $G_i$; each ordering in $\Sigma_i$ is a partial ordering of $G$. We then set $\Sigma \leftarrow \Sigma'_T\bigcup (\cup_{i\in [1,s]}\Sigma_i)$. This completes the construction of $\Sigma$.

	Next, we argue that $\Sigma$ is an $(O(\frac{k}{\epsilon}\cdot\log n),1)$-left-sided LSO. Since the recursion depth is $k$, and each vertex belongs to at most $O(\frac{1}{\epsilon}\cdot\log n)$ partial orderings in each recursive level, it follows that each vertex belongs  to at most $O(\frac{k}{\epsilon}\cdot\log n)$ partial orderings in $\Sigma$. 
	
	Finally, consider a pair of vertices $u,v\in V$. Let $P_{u,v}$ be a shortest path from $u$ to $v$ in $G$.
	Let $G'$ be the subgraph of $G$ at highest depth in the recursive construction that contains all the vertices of $P_{u,v}$.
	In particular $d_{G'}(v,u)= d_G(v,u)$.
	Let $P$ be the shortest path deleted from $G'$. Let $T$ be the tree constructed for $P,G'$. 
	By maximality $P\cap P_{u,v}\not= \emptyset$, thus there are two copies  $v_i,u_j$ of $v,u$ such that $d_{T}(v_i,u_j)\le(1+\epsilon)\cdot d_{G'}(v,u)=(1+\epsilon)\cdot d_G(v,u)$.
	By \Cref{lem:SeperatorToLSO}, $\Sigma_T$ contains an ordering  $\sigma$ such that $v_i,u_j\in \sigma$, 
	and for every $x\preceq_\sigma v_i$, $y\preceq_\sigma u_j$ it holds that $d_T(x,y)\le d_T(v_i,u_j)$.
	The collection of orderings $\Sigma'_T$ contains an ordering $\sigma'$, which is simply $\sigma$ with removed duplicates. In particular, $\sigma'$ contains only original vertices from $V$.
	Let $v_{i'},u_{j'}$ be the leftmost duplicates of $v,u$ (respectively) in $\sigma$.
	For every $x,y\in V$ such that $x\preceq_{\sigma'} v$, $y\preceq_{\sigma'}u$, there are copies $x_p,y_q$ of $x,y$ such that $x_p\preceq_\sigma v_{i'}\preceq_\sigma v_i$ and $y_q\preceq_\sigma u_{j'}\preceq_\sigma u_j$. It thus holds that 
	\[
	d_{G}(x,y)\le d_{T}(x_{p},y_{q})\le d_{T}(v_{i},u_{j})\le(1+\eps)\cdot d_{G}(v,u)~.
	\]
	This completes the proof of~\Cref{thm:SpannerForMinorFree}.
\end{proof}

\subsection{Reliable Spanners from Left-sided LSO}\label{subsec:LeftLSOToSpanners}

To create a reliable spanner from a left-sided LSO, we will use \emph{left spanners}:
\begin{definition}[Reliable Left Spanner]\label{def:RelLeftSpanner}
	Given a path graph $P_n$, a $2$-hop left spanner $H$ is a graph such that for every $a<b$ there is $c\le a,b$, such that $\{v_c,v_a\},\{v_c,v_b\}\in H$.
	A distribution over $2$-hop left spanners $H$ is 
	$\nu$-reliable if for every attack $B\subseteq X$, there is a set $B\subseteq B^+$ such that $\mathbb{E}[B^+]\le(1+\nu)|B|$, and for every $a<b$ such that $v_a,v_b\notin B^+$, there is $c\le a,b$, $c\notin B$ such that $\{v_c,v_a\},\{v_c,v_b\} \in E(H)$.
	An oblivious spanner $\mathcal{D}$ is said to have $m$ edges if every spanner in the support $\mathcal{D}$ has at most $m$ edges (see \Cref{lem:ExpectedSize}).
\end{definition}
In the following lemma, we construct a reliable left spanner. Interestingly, it is sparser than the reliable $2$-hop spanner for the path graph constructed in \Cref{lem:2-hopSpanner}. The proof is deferred to \Cref{subsec:2HopLeftSpanner}.
\begin{restatable}[$2$-hop-left-spanner]{lemma}{LeftSpanner}
	\label{lem:2HopLeft}
	For every $\nu\in(0,1)$, the $n$-vertex path $P_n$ admits an oblivious $\nu$-reliable $2$-hop left spanner with $n\cdot O(\nu^{-1}\log n)$ edges.	
\end{restatable}

Using \Cref{lem:2HopLeft}, we can construct a reliable spanner using a left-sided LSO, and thus prove \Cref{thm:MetaLeftLSOMain}. The meta theorem is restated bellow for convenience.

\SpannerLeftLSO*
\begin{proof}
	Let $\Sigma$ be a $(\tau,\rho)$-left-sided LSO; $\Sigma$ is a collection of partial orderings. Set $\nu'=\frac{\nu}{\tau}$. For every ordering $\sigma\in \Sigma$ which contains $n_\sigma$ vertices, using \Cref{lem:2HopLeft}, we construct a $\nu'$-reliable $2$-hop left spanner $H_\sigma$ with $n_\sigma\cdot O(\nu'^{-1}\log n_\sigma)$ edges. 
	Set $H=\cup_\sigma H_\sigma$.
	The total number of edges is thus bounded by
	\[
	|H|=\sum_{\sigma}|H_{\sigma}|=\sum_{\sigma}n_{\sigma}\cdot O(\nu'^{-1}\log n_{\sigma})=O(\nu'^{-1}\log n)\cdot\sum_{\sigma}n_{\sigma} = n\cdot O(\nu^{-1}\tau^2\log n)~.
	\]
	For an attack $B$, let $B^+_\sigma$ be the faulty extension of $B$ w.r.t. $H_\sigma$, and let $B^+=\cup B^+_\sigma$.
	For every pair of points $x,y\notin B^+$, let $\sigma$ be the promised left-sided partial ordering  from \Cref{def:leftsidedLSO}. As $x,y\notin B_\sigma^+$, there is some vertex $z\notin B$ such that $z\preceq_{\sigma}x,y$ and $(x,z),(y,z)\in E(H_\sigma)$. We conclude that 
	\[
	d_{H}(x,y)\le d_{X}(x,z)+d_{X}(z,y)\le2\rho\cdot d_{X}(x,y)~.
	\]
	To bound the size of $B^+$, we observe that
	\[
	\mathbb{E}\left[\left|B^{+}\right|\right]\le|B|+\sum_{\sigma}\mathbb{E}\left[\left|B_{\sigma}^{+}\setminus B\right|\right]\le|B|+\tau\nu'\cdot|B|\le(1+\nu)\cdot|B|~.
	\]
	The theorem now follows.
\end{proof}

\subsubsection{Left Spanner: proof of \Cref{lem:2HopLeft}}\label{subsec:2HopLeftSpanner}
We restate the lemma for convenience.
\LeftSpanner*
\begin{proof}	
	We will construct a spanner with reliability parameter $O(\nu)$; afterward, the constants can be readjusted accordingly. 
	Let $h$ be a universal constant such that $H_{s}-H_{t}=\frac{1}{t+1}+\frac{1}{t+2}+\dots+\frac{1}{s}>h\cdot\ln\frac{s}{t}$.
	We create a set $N$ of centers as follows: sample each vertex $a\in[n]$ to $N$ with probability $\frac{1}{a}\cdot\frac{c}{\nu}$ (or $1$ for small enough $a$) where $c=\max\{\frac2h,1\}$.
	The spanner $H$ is then defined as follows: for each vertex $a\in V$ set $N_a=[1,a]\cap N$. We add edges between $a$  and all the vertices in $N_a$ (alternatively, each vertex add edges to all the centers with smaller indices).	
	Clearly, the expected number of edges is bounded by
	\[
	\sum_{a=1}^{n}\frac{1}{a}\cdot\frac{c}{\nu}\cdot(n-a)\le n\cdot\frac{c}{\nu}\cdot\sum_{a=1}^{n}\frac{1}{a}=n\cdot O(\nu^{-1}\log n)~.
	\]
	By \Cref{lem:ExpectedSize}, we can obtain the same guarantee on the number of edges also in the worst case.	
	
	Given an attack $B$, we add a vertex $a$ to $B^+$ iff $N_a\subseteq B$. The stretch bound easily follows: for a pair of vertices $a<b\notin B^+$ there is a center $x\in N_a\subseteq N_b$, thus $\{x,a\},\{x,b\}\in H$ as required.
	
	Finally we bound the expected size of $B^+$ for an attack $B$.
	We will begin by analyzing the specific attack $B=[1,x]$ for some $x\in[n]$. Later, we will argue that our analysis holds also for an arbitrary attack $B$. 
	The vertex $x+i$ joins $B^{+}$ only if none of the vertices $x+1,\dots,x+i$
	becomes a center in $N$. The probability of this event is 
	\[
	\Pr\left[x+i\in B^{+}\right]=\Pi_{j=1}^{i}\left(1-\frac{1}{x+j}\cdot\frac{c}{\nu}\right)<e^{-\frac{c}{\nu}\sum_{j=1}^{i}\frac{1}{x+j}}<e^{-\frac{c}{\nu}\cdot h\left(\ln\frac{x+i}{x}\right)}\le(1+\frac{i}{x})^{-\frac{2}{\nu}}\,.
	\]
	The expect size of $B^{+}\setminus B$ is thus bounded by
	\[
	\mathbb{E}\left[\left|B^{+}\setminus B\right|\right]<\sum_{i=1}^{n-x}(1+\frac{i}{x})^{-\frac{2}{\nu}}<\sum_{i\ge1}(1+\frac{i}{x})^{-\frac{2}{\nu}}
	\]
	For every $s\ge0$, and $i\in[s\nu x,(s+1)\nu x)$, we
	have that $(1+\frac{i}{x})^{-\frac{2}{\nu}}<(1+s\nu)^{-\frac{2}{\nu}}<e^{-\frac{s\nu}{2}\cdot\frac{2}{\nu}}=e^{-s}$.
	We conclude: 
	\[
	\mathbb{E}\left[\left|B^{+}\setminus B\right|\right]\le\sum_{s\ge0}\sum_{i\in[s\nu x,(s+1)\nu x)}(1+\frac{i}{x})^{-\frac{2}{\nu}}<\nu x\sum_{s\ge0}e^{-s}=O(\nu)\cdot|B|~.
	\]
	It follows that $\mathbb{E}\left[\left|B^{+}\right|\right]=(1+O(\nu))|B|$ as required.
	Finally, consider an arbitrary attack $B\subseteq [n]$ of size $x$. Let $a_1,a_2,\dots,a_{n-x}$ the vertices not in $B$, ordered from left to right. Note that $a_i$ has $i$ vertices to the left (or equal) of it not in $B$, all from the interval $[1,x+i]$. To bound $\Pr\left[a_{i}\in B^{+}\right]$, we observe that this probability is maximum when $B = [1,x]$ since the probability that vertices are sampled to $N$ decreases monotonically. Thus,  $\Pr\left[a_{i}\in B^{+}\right]\le\Pi_{j=1}^{i}\left(1-\frac{1}{x+j}\cdot\frac{c}{\nu}\right)$. The rest of the analysis of $|B^+|$ follows exactly the same argument.
\end{proof}

%% file: Subgraph.tex
\section{Subgraph Reliable Connectivity Preserves}\label{sec:subgraphConnectivity}

Classically, throughout the literature, given a graph $G=(V,E,w)$ a spanner $H$ of $G$ is required to be a subgraph of $G$. A natural question in the context of this paper is: Is it possible to construct a reliable spanner with a subquadratic number of edges that only uses edges of the input graph? 
Since removing a single vertex could disconnect the graph into two equal parts, a necessary relaxation is to require that distances are preserved only w.r.t. the induced graph $G[V\setminus B]$ (similarly to the case of fault-tolerant spanners).   
Here we show that even for the less ambitious task of constructing a $\nu$-reliable connectivity preserver, it is not possible without using $\Omega(n^2)$ edges.  To this end, we introduce a generalized definition of reliability, and present an oblivious lower bound and a matching deterministic upper bound.

\begin{definition}[Reliable Connectivity Preserver]\label{def:preserver}
	Given a graph $G=(V,E)$, and a monotonically non-decreasing function $g:\mathbb{N}\rightarrow\mathbb{N}$, a subgraph $H$ of $G$ is a deterministic $g$-reliable connectivity preserver if for every attack $B\subseteq V$, there is a superset $B\subseteq B^{+}$
	of size at most $g(|B|)$, such that for every $u,v\in V\setminus B^{+}$, if $u$ and $v$ are connected in $G\setminus B$, then they are also connected in $H\setminus B$.
	
	An oblivious $g$-reliable connectivity preserver is a distribution $\mathcal{D}$ over subgraphs $H$ of $G$, such that for every attack $B$ and $H\in \supp(\mathcal{D})$, there is an algorithm producing a superset $B^{+}$ of $B$ such that, for
	every $u,v\notin B^{+}$, if $u$ and $v$ are connected in $G\setminus B$, then they are also connected in $H\setminus B$,
	and furthermore, $\mathbb{E}_{H\sim\mathcal{D}}\left[|B^{+}|\right]\le g(|B|)$. We say that an oblivious $g$-reliable connectivity preserver $\mathcal{D}$ has $m$ edges if every graph $H$ in the support of $\supp(\mathcal{D})$ has at most $m$ edges.
\end{definition}
Note that using the notation of \Cref{def:preserver}, the rest of the paper is concerned with constructing oblivious $g$-reliable (non-subgraph) spanners for the linear function $g(x)=(1+\nu)x$.
We begin with the upper bound.
\begin{theorem}\label{thm:preserverUB}
	Let $k>1$ be an integer and $g_{k}(x)= \Omega(x^k)$. Then every $n$-vertex graph admits a deterministic $g_{k}$-reliable connectivity preserver with $O(n^{1+1/k})$ edges.
\end{theorem}
\begin{proof}
	The construction is by reduction to $f$-fault-tolerant spanners. Recall that a subgraph $H$ of $G$ is an  $f$-vertex-fault-tolerant $t$-spanner if for every subset $B$ of at most $f$ vertices, for every $u,v\in V\setminus B$, $d_{H\setminus B}(u,v)\le t\cdot d_{G\setminus B}(u,v)$. Bodwin and Patel \cite{BP19Spanners} constructed $f$-vertex-fault-tolerant $2t-1$-spanners with $O(n^{1+\frac1t}\cdot f^{1-\frac1t})$ edges (later Bodwin, Dinitz, and Robelle \cite{BDR21} showed how to construct such spanners efficiently). Assume that $g(x)\geq cx^k$ for some constant $c$.	Set $f=(n/c)^{\frac{1}{k}}$, and construct an $f$-vertex-fault-tolerant $O(\log n)$-spanner $H$ with $O(n\cdot f)=O(n^{1+1/k})$ edges.
	We argue that $H$ is a deterministic $g_{k}$-reliable connectivity preserver. Clearly $H$ has $O(n^{1+1/k})$ edges. Consider an attack $B\subseteq V$,
	\begin{itemize}
		\item 	If $|B|\le f$, set $B^{+}=B$. Clearly $|B^+|\le g(|B|)$, and for every $u,v\in V\setminus B^+$, if $u,v$ are connected in $G\setminus B$, then $d_{H\setminus B}(u,v)\le t\cdot d_{G\setminus B}(u,v)<\infty$, and in particular $u$ and $v$ are connected in $H\setminus B$.
		\item Else, $|B|>f$, set $B^{+}=V$. Then $|B^+|=n= c((n/c)^{\frac1k})^k\le g(|B|)$. The connectivity preservation holds trivially.
	\end{itemize}
\end{proof}
Note that in \Cref{thm:preserverUB} we actually constructed a $g_k$-reliable deterministic $O(\log n)$ spanner. One can also observe that by the same construction, we get the following corollary.

\begin{corollary}\label{cor:subgraphSpannerrUB}
	Let $k>1$ be an integer and $g_{k}(x)= \Omega(x^k)$. Then every graph $n$-vertex
	graph $G=(V,E,w)$ admits a deterministic $g_{k}$-reliable $(2t-1)$-spanner with	$O(n^{1+\frac{t+k-1}{t\cdot k}})$ edges that only uses edges of $G$. 
\end{corollary}

Next, we provide a matching lower bound for the connectivity preserver; our lower bound is actually  stronger as  it holds in the oblivious case.

\begin{theorem}\label{thm:preserverLowerbound}
		Let $k>1$ be an integer and $g_{k}(x)= O(x^k)$. For every large enough  $n$, there is an $n$-vertex graph $G$ such that every oblivious $g_{k}$-reliable connectivity preserver has $\Omega(n^{1+1/k})$	edges.
	edges.
\end{theorem}
\begin{proof} 
	For simplicity, we show the lower bound for $g_k(x) = x^k$; the lower bound could be extended to any function  $g_{k}(x)\leq c_0\cdot x^k$ for some constant $c_0$ by adjusting the constants in our proof.
	We define a graph $G$ to be a ``thick'' cycle, specifically, for $c=12$ we will have $s=cn^{1-\frac{1}{k}}$ sets $A_0,A_1,\dots,A_{s-1}$, each containing $\frac{1}{c}n^{1/k}$ vertices. There will be an edge between $v\in A_i$ to $u\in A_j$ if and only if $i-j=1 (\mod s)$.
	All the index computation in the rest of the proof will be done modulo $s$, and we will omit it in the notation.
	We will also ignore rounding issues (which can be easily fixed).
	Note that $|V|=n$, while $|E|=cn^{1-\frac{1}{k}}\cdot\frac{1}{c^{2}}n^{\frac{2}{k}}=\frac{1}{c}n^{1+1/k}$. See \Cref{fig:ThickCycle} for illustration.	
	
	We begin by proving that every deterministic reliable connectivity preserver has $\Omega(n^{1+1/k})$ edges. Later, we will generalize to the oblivious case.
	We argue that only the graph $G$ itself is a deterministic $g_{k}$-reliable connectivity preserver. Suppose for contradiction that there is a $g_{k}$-reliable connectivity preserver $H$ missing the
	edge $\{u,v\}$, where $u\in A_i$ and $v\in A_{i+1}$. Consider an attack $B=V_i\cup V_{i+1}\cup V_{i+\frac s2}\setminus \{u,v\}$. Note that $|B|<\frac{3}{c}\cdot n^{1/k}$.
	Note that the graph $G\setminus B$ is connected.
	The graph $G'$ obtained by removing the edge $\{u,v\}$ from $G\setminus B$ has two connected components $C_1=\{v\}\cup A_{i+2}\cup \cdots\cup V_{i+\frac s2-1}$, and $C_2=V_{i+\frac s2-1}\cup\cdots \cup V_{i-1}\cup\{u\}$, see \Cref{fig:ThickCycle} for illustration.
	The set $B^+$ must contain all the vertices in either $C_1$ or $C_2$, as otherwise there will vertices connected in $G\setminus B$, but not connected in $H\setminus B\subseteq G'\setminus B$.
	It follows that $|B^+|\ge\min\{|C_1\cup B|,|C_2\cup B|\}\ge\frac n2$.
	But $g_{k}(|B|)|=\left(\frac{3}{c}\cdot n^{1/k}\right)^{k}=\left(\frac{3}{c}\right)^{k}\cdot n<\frac{n}{4}\le|B^+|$,
	a contradiction. It follows that every deterministic $g_{k}$-reliable
	connectivity preserver has $\Omega(n^{1+1/k})$ edges.
	
	Next we generalize the lower bound to the oblivious case. Let $\mathcal{D}$ be a distribution over connectivity preservers, and suppose
	that there is an edge $\{u,v\}$ as above such that $\Pr_{H\sim\mathcal{D}}\left[e\in H\right]<\frac{1}{2}$.
	Then using the same argument as above, with the same attack $B$ (w.r.t. $\{u,v\}$) we get that in all the preservers $H\in\supp(\mathcal{D})$ not containing $\{u,v\}$, it holds that  $|B^+|\ge\frac n2$. We conclude
	\[
	\mathbb{E}_{H\sim\mathcal{D}}[|B^{+}|]\ge\text{\ensuremath{\Pr\left[e\notin H\right]\cdot}}\mathbb{E}_{H\sim\mathcal{D}}[|B^{+}|\mid e\notin H]\ge\frac{1}{2}\cdot\frac{n}{2}~,
	\]
	As $g_{k}(|B|)<\frac{n}{4}$, it follows that the probability that $H$ contains an arbitrary edge  $\{u,v\}$ is at least $\frac12$. In particular 
	$\mathbb{E}_{H\sim\mathcal{D}}\left[|H|\right]\ge\frac{1}{2}|G|=\Omega(n^{1+1/k})$. The theorem now follows.
\end{proof}
\begin{figure}[t]
	\centering
	\includegraphics[width=0.8\textwidth]{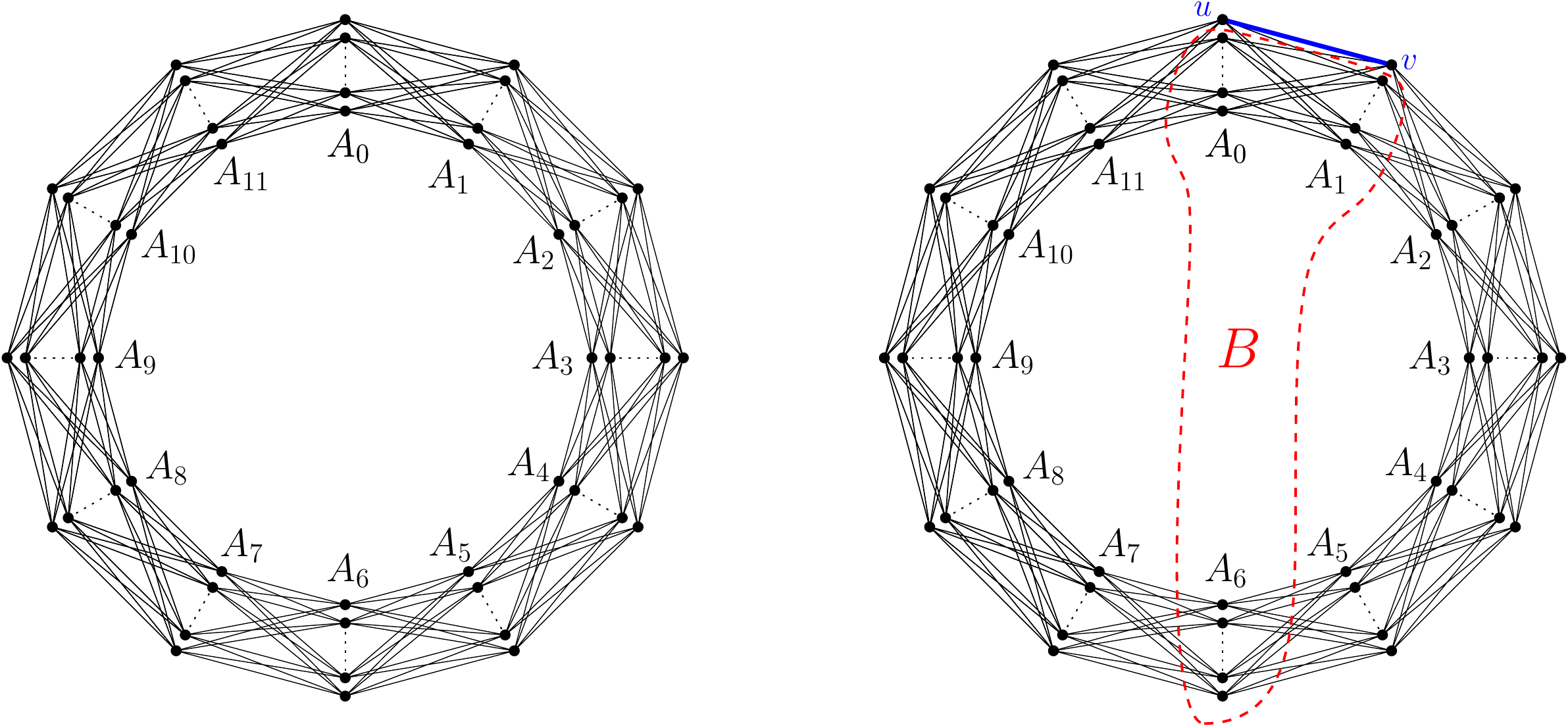}
	\caption{\footnotesize{On the left illustrated the thick cycle consisting of $12$ sets $A_0,\dots,A_1$, where there is complete bipartite graph between $A_i$ to $A_{i+1}$.
			On the right, illustrated the graph $G'$ lacking an edge $\{u,v\}$ colored in blue. The attack $B=V_i\cup V_{i+1}\cup V_{i+\frac s2}\setminus \{u,v\}$ (encircled in red) disconnects the graph $G'$ into two nearly equal halves, while keeping $G$ connected.
	}}
	\label{fig:ThickCycle}
\end{figure}

By setting $g(x) = (1+\nu)x$, we obtain a lower bound $\Omega(n^2)$ on the number of edges of any oblivious $\nu$-reliable \emph{subgraph} connectivity preserver, and hence any oblivious $\nu$-reliable \emph{subgraph} $t$-spanner for any finite $t$.